%% file: KTY_Barrier_Full.tex
\documentclass[a4paper,12pt]{article}

\makeatletter
 
  \@addtoreset{equation}{section}
 \makeatother
\usepackage{amsfonts}
\usepackage{amssymb}
\usepackage{mathrsfs}
\usepackage{amsmath,amsthm}
\usepackage{amsmath}

\usepackage{graphicx}
\usepackage{color}
\usepackage{indentfirst}

\begin{document}
\input{tymacro_3}

\newtheorem{thm}{Theorem}
\newtheorem{lemma}{Lemma}
\newtheorem{prop}{Proposition}
\newtheorem{defn}{Definition}
\newtheorem{rem}{Remark}
\newtheorem{step}{Step}
\newtheorem{cor}{Corollary}

\newcommand{\Cov}{\mathop {\rm Cov}}
\newcommand{\Var}{\mathop {\rm Var}}
\renewcommand{\E}{\mathop {\rm E}}
\newcommand{\const }{\mathop {\rm const }}
\everymath {\displaystyle}

\newcommand{\ruby}[2]{
\leavevmode
\setbox0=\hbox{#1}
\setbox1=\hbox{\tiny #2}
\ifdim\wd0>\wd1 \dimen0=\wd0 \else \dimen0=\wd1 \fi
\hbox{
\kanjiskip=0pt plus 2fil
\xkanjiskip=0pt plus 2fil
\vbox{
\hbox to \dimen0{
\small \hfil#2\hfil}
\nointerlineskip
\hbox to \dimen0{\mathstrut\hfil#1\hfil}}}}

\def\qedsymbol{$\blacksquare$}
\renewcommand{\thefootnote }{\fnsymbol{footnote}}
\renewcommand{\refname }{References}

\everymath {\displaystyle}

\title{
A Semi-group Expansion for Pricing Barrier Options
\footnote{We are very grateful to Professor Seisho Sato in
the University of Tokyo
for his suggestions to our numerical computation.}}
\date{First Version: February, 2012, This Version: May 30, 2013}
\author{Takashi Kato\footnote{
Division of Mathematical Science for Social Systems, Graduate School of Engineering Science,
Osaka University, 1-3, Machikaneyama-cho, Toyonaka, Osaka 560-8531, Japan, 
E-mail: \texttt{kato@sigmath.es.osaka-u.ac.jp}
}
\and Akihiko Takahashi\footnote{
Graduate School of Economics, The University of Tokyo, 
7-3-1 Hongo, Bunkyo, Tokyo, 113-0033 Japan. 
}
\and 
Toshihiro Yamada\footnote{
Mitsubishi UFJ Trust Investment Technology Institute Co.,Ltd. (MTEC), 
2-6, Akasaka 4-Chome, Minato, Tokyo, 107-0052 Japan, 
E-mail: \texttt{yamada@mtec-institute.co.jp}
} \footnote{
Graduate School of Economics, The University of Tokyo, 
7-3-1 Hongo, Bunkyo, Tokyo, 113-0033 Japan. 
} }
\maketitle

\allowdisplaybreaks
\begin{abstract}
This paper presents a new asymptotic expansion method
for pricing continuously monitoring barrier options.
In particular, we develops a semi-group expansion scheme for
the Cauchy-Dirichlet problem in the second-order parabolic partial differential equations (PDEs)
arising in
barrier option pricing.
As an application, we propose a concrete approximation formula 
under a stochastic volatility model
and demonstrate its validity by some numerical experiments.
\footnote[0]{Mathematical Subject Classification (2010) \  35B20, 35C20, 91G20} 
\end{abstract}

\section{Introduction}\label{sec_intro}
Since the Merton's seminal work (\cite{Merton})
barrier options have been quite popular and important products
in both academics and financial business for the last four decades.
In particular, fast and accurate computation of their prices and Greeks is highly desirable in the risk management,
which is a tough task under the finance models commonly used in practice.
Thus, it has been one of the central issues in the mathematical finance community.
Among various approaches to attacking the problem, this paper proposes a new semi-group expansion scheme 
under general diffusion setting.

Firstly, let us note that the value of a continuously monitoring down-and-out  barrier option is expressed as the following form:
\begin{eqnarray}\label{def_CB}
C_\mathrm {Barrier}(T,x) = 
\E[f(X_T^x)1_{\{ \tau > T \}}] =\E[f(X_T^x)1_{\{ \min_{t \in [0,T]} X_{t} > L \}}]. 
\end{eqnarray} 
Here, $T > 0$ is a maturity of the option, and
$(X_t^x )_t$ denotes a vector process
starting from $x$ including a price process of the underlying asset
(usually given as the solution of a certain stochastic differential equation (SDE)).
Also, $L$ stands for a constant lower barrier, that is
$L < x$, and
$\tau$ is the hitting time to $L$:
\begin{eqnarray}
\tau = \inf \{ t\in [0, T] : X_t^x \leq  L  \}.
\end{eqnarray} 

It is well-known that a possible approach in computation of 
$C_\mathrm {Barrier}(T,x)$
is the Euler-Maruyama scheme, which 
stores the sample paths of the process $(X^x_t)_t$ through an $n$-time discretization 
with the step size $T / n$. 
When applying this scheme to pricing a continuously monitoring barrier option,
one kills the simulated process, say $(\bar{X}^x_{t_i})_i$
if $\bar{X}^x_{t_i}$ exits from the domain $(L, \infty )$ until the maturity $T$.  
The usual Euler-Maruyama scheme is {\it suboptimal} 
since it does not control the diffusion paths between 
two successive dates $t_i$ and $t_{i+1}$: 
the diffusion paths could have crossed the barriers and come back to the domain without being detected. 
It is also known that the error 
between
$C_\mathrm {Barrier}(T,x)$ and 
$\bar{C}_\mathrm {Barrier}(T,x)$, the barrier option price obtained by the Euler-Maruyama scheme
is of order $\sqrt{T / n}$,
as opposed to 
the 
order $T / n$ for standard plain-vanilla options. (See \cite{Gobet}) 
Thus, 
various Monte-Carlo schemes 
have been proposed for improving the order of the error. 
(See \cite {Pham} for instance.)

One of the other tractable approaches for 
calculating $C_\mathrm {Barrier}(T,x)$
is to derive an analytical approximation.
If we obtain an accurate 
approximation formula, 
it is a powerful tool for pricing continuously monitoring barrier options 
because we need not rely on Monte-Carlo simulations anymore.
However, from a mathematical viewpoint, 
deriving an approximation formula by applying stochastic analysis 
is not an easy task since the Malliavin calculus cannot be directly applied.
It is due to 
the non-existence of the Malliavin derivative $D_t\tau$ (see \cite {FLLL}) 
and to the fact that 
the minimum (maximum) process of the Brownian motion has only the first-order differentiability in the Malliavin sense. 
Thus, neither approach in \cite{Kunitomo-Takahashi} nor in \cite{Takahashi-Yamada} 
can be applied directly to valuation of
continuously monitoring barrier options, while they are applicable to 
pricing discrete barrier options. (See \cite {STY1} for the detail.)

This paper proposes a new general method for 
the approximation of barrier option prices. 
Particularly, our objective is to pricing barrier options when the dynamics of
the underlying asset price is described by the following perturbed SDE:  
\begin{eqnarray}\label{eq_SDE_epsilon0}
\left\{
\begin{array}{l}
 	dX^{\varepsilon , x}_t = b(X^{\varepsilon , x}_t, \varepsilon )dt + \sigma (X^{\varepsilon , x}_t, \varepsilon )dB_t,\\
 	X^{\varepsilon , x}_0 = x, 
\end{array}
\right.
\end{eqnarray}
where $\varepsilon$ is a small parameter, which will be defined precisely in the next section. 
In this case, the barrier option price (\ref {def_CB}) is 
characterized as a solution of the Cauchy-Dirichlet problem: 
\begin{eqnarray}\label{PDE_intro}
\left\{
\begin{array}{ll}
 	\frac{\partial }{\partial t}u^\varepsilon (t, x) + \mathscr {L}^\varepsilon u^\varepsilon (t, x) = 0,& (t, x)\in [0, T)\times (L, \infty ),	\\
 	u^\varepsilon (T, x) = f(x),& x > L,	\\
 	u^\varepsilon (t, L) = 0,& t\in [0, T], 
\end{array}
\right.
\end{eqnarray}
where the differential operator $\mathscr {L}^\varepsilon $ is determined
by the diffusion coefficients $b$ and $\sigma $. 
Next, we introduce an asymptotic expansion formula: 
\begin{eqnarray}\label{ae_temp}
u^{\varepsilon }(t, x) = u^0(t, x) + \varepsilon v^0_1(t, x) + \cdots  + \varepsilon ^{n-1}v^0_{n-1}(t, x) + 
O(\varepsilon ^n), 
\end{eqnarray}
where $O$ denotes the Landau symbol. 
The function $u^0(t, x)$ is the solution of (\ref {PDE_intro}) with $\varepsilon = 0$: 
if $b(x, 0)$ and $\sigma (x, 0)$ have some simple forms such as constants
(as in the Black-Scholes model), 
we already know the closed form of $u^0(t, x)$ and hence obtain the price.
Then, we are able to get the approximate value for $u^{\varepsilon }(t, x)$ 
through evaluation of the coefficient functions
$v^0_1(t, x), \ldots , v^0_{n-1}(t, x)$. 
In fact, they are also characterized as the solution of 
a certain PDE with the Dirichlet condition.  
By formal asymptotic expansions, 
(\ref {ae_temp}) above and (\ref{ae_generator}) below,
\begin{eqnarray}\label{ae_generator}
\mathscr {L}^{\varepsilon } = 
\mathscr {L}^0 + \varepsilon \tilde{\mathscr {L}}^0_1 + \cdots  + 
\varepsilon ^{n-1}\tilde{\mathscr {L}}^0_{n-1} + \cdots , 
\end{eqnarray}
we can derive the following PDE which $v^0_k(t, x)$ satisfies:
\begin{eqnarray}\label{PDE_v_intro}
\left\{
\begin{array}{ll}
 	\frac{\partial }{\partial t}v^0_k(t, x) + \mathscr {L}^0v^0_k(t, x) + g^0_k(t, x) = 0,& (t, x) \in [0, T)\times (L, \infty ),	\vspace{1mm}\\
 	v^0_k(T, x) = 0,& x > L,	\\
 	v^0_k(t, L) = 0,& t\in [0, T], 
\end{array}
\right.
\end{eqnarray}
where $g^0_k(t, x)$ will be given explicitly in Section 3. 
Moreover, by applying the Feynman-Kac approach to the PDE (\ref{PDE_v_intro}), 
we obtain a semi-group representation of $v^0_k$. That is,
for each $k = 1, \ldots , n - 1$, 
\begin{eqnarray}
&&v^0_k(T - t, x)\nonumber\\
&& = \sum ^k_{l = 1}\sum _{(\beta ^i)^l_{i=1}\subset \Bbb {N}^l, \sum _i\beta ^i = k}
\int ^t_0\int ^{t_1}_0\cdots \int ^{t_{l - 1}}_0
P^D_{t - t_1}\tilde{\mathscr {L}}^0_{\beta ^1}P^D_{t_1 - t_2}\tilde{\mathscr {L}}^0_{\beta ^2}\cdots P^D_{t_{l - 1} - t_l}\tilde{\mathscr {L}}^0_{\beta ^l}P^D_{t_l}f(x)dt_l\cdots dt_1,\nonumber\\
\label{intro_semigroup_formula} 
\end{eqnarray}
where $(P^D_t)_t$ is a semi-group defined in Section 3.
We will justify the above argument in a mathematically rigorous manner
in the sections 2 and 3.

The theory of the Cauchy-Dirichlet problem for this kind of the second order parabolic PDE 
is well understood for the case of bounded domains 
(see \cite{Friedman1}, \cite{Friedman2} and \cite{Lieberman} for instance). 
As for an unbounded domain case such as (\ref {PDE_intro}), 
\cite {Rubio} 
provides the existence and uniqueness results for a solution of the PDE and 
the Feynman-Kac type formula, the part of which will be cited as Theorem \ref {th_viscosity} in Section 2. 
However, some mathematical difficulty exists for
applying the results of \cite {Rubio} to  the PDE (\ref {PDE_v_intro}).
More precisely, the function $g^0_k(t, x)$ may be divergent at $t = T$. 
Hence, in order to obtain an asymptotic expansion (\ref{ae_temp}),
we generalize the result of
 \cite {Rubio} and the argument of the Feynman--Kac representation.
Furthermore, 
we derive a new  representation (\ref{intro_semigroup_formula}) for $v^0_k(t, x)$ 
by using the semi-group $(P^D_t)_t$.
We notice that such a form is convenient for evaluation of $v^0_k(t, x)$ in concrete examples. 

We also apply our method to pricing a barrier option in a stochastic volatility model. 
 
Then, as an example of (\ref{intro_semigroup_formula}) we obtain a new approximation formula 
of the barrier option price $C_\mathrm {Barrier}^{SV,\varepsilon}$ under a stochastic volatility model 
as follows: for the initial value of the logarithmic underlying price $x$, the maturity $T$ and the lower barrier $L$,  
\begin{eqnarray*}
C_\mathrm {Barrier}^{SV,\varepsilon}(T,e^x)&=& \E \left[ f(S^{\varepsilon , x}_T )1_{\{\min _{0\leq t\leq T}S^\varepsilon _t > L\}}\right]\\
&\simeq &
P^{D}_{T}\bar{f}(x)+\varepsilon \int ^T_0P^{D}_{T - r}\tilde{\mathscr {L}}^0_1P^{D}_r \bar{f}(x)dr, 
\end{eqnarray*}
where $(S^{\varepsilon , x}_t)_t$ 
is the underlying asset price process, 
$f$ is a payoff function and $\bar{f}(x)=f(e^x)$. 
Here, $P^{D}_{T}\bar{f}(x)$ is regarded as the down-and-out barrier option price
in the Black-Scholes model. 
Moreover, we confirm practical validity of our method through a numerical example given in Section 4.

Finally, we remark that
there exist the previous works on barrier option pricing 
such as \cite{FPS1}, \cite{FPS2}, \cite {HS}, \cite{IJS},
which start with some specific models
(e.g. Black-Scholes model or some type of fast mean-reversion model),
and derive approximation formulas for discretely or continuously monitoring  barrier option prices.
Our approach is to firstly develop a general 
semi-group expansion scheme for the Cauchy-Dirichlet problem 
under multi-dimensional diffusion setting; then as an application,
we provide a new approximation formula under a certain class of stochastic volatility model.

The organization of this paper is as follows:
the next section prepares the existence and uniqueness result
for the Cauchy-Dirichlet problem
in the second-order parabolic PDE associated with barrier option pricing.
Section 3 presents our main result for an asymptotic expansion of barrier option prices.
Section 4 shows numerical examples under a 
 stochastic volatility model.
Section 5 concludes.
Finally, Appendix A provides the proofs of the results in the main text, and Appendix B
shows some generalization of Section 2 and Section 3.  
\section{Preparation}
\label{sec_main}
This section shows the existence and uniqueness result for the Cauchy-Dirichlet problem
in the second-order parabolic PDE associated with the valuation of barrier options.

Suppose first that the underlying asset price is described by the following perturbed SDE:  
\begin{eqnarray}\label{eq_SDE_epsilon}
\left\{
\begin{array}{l}
 	dX^{\varepsilon , x}_t = b(X^{\varepsilon , x}_t, \varepsilon )dt + \sigma (X^{\varepsilon , x}_t, \varepsilon )dB_t,\\
 	X^{\varepsilon , x}_0 = x, 
\end{array}
\right.
\end{eqnarray}
where $\varepsilon$ is a small parameter. 
Let $b : \Bbb {R}^d\times  I \longrightarrow \Bbb {R}^d$ and 
$\sigma : \Bbb {R}^d\times  I  \longrightarrow \Bbb {R}^d\otimes \Bbb {R}^m$ be 
Borel measurable functions ($d, m\in \Bbb {N}$,) where $I$ is an interval on $\Bbb {R}$ including the origin $0$ 
(for instance $I = (-1, 1)$.) 
We consider the SDE (\ref {eq_SDE_epsilon})
for any $x\in \Bbb {R}^d$ and $\varepsilon \in  I $; 
in the condition [A] below,
we will introduce the assumptions for existence and uniqueness of a solution of (\ref {eq_SDE_epsilon}).

\enlargethispage{5mm}
We are interested in evaluation of the following barrier option price: for a small $\varepsilon $,
\begin{eqnarray}\nonumber 
&&u^{\varepsilon }(t, x)\\\label{def_u_eps}
&=& \E \left [\exp \left( -\int ^{T-t}_0c(X^{\varepsilon , x}_r, \varepsilon )dr 
\right) f(X^{\varepsilon , x}_{T-t})1_{\{ \tau _D(X^{\varepsilon , x}) \geq T-t \}}\right ], \ \ (t, x)\in [0, T]\times \bar{D}\hspace{10mm}
\end{eqnarray}
for Borel measurable functions $f : \Bbb {R}^d \longrightarrow \Bbb {R}$ and 
$c : \Bbb {R}^d\times  I  \longrightarrow \Bbb {R}$, a positive real number $T > 0$ 
and a domain $D\subset \Bbb {R}^d$; 
$\bar{D}\subset \Bbb {R}^d$ is the closure of $D$ and 
$\tau _D(w)$, $w\in C([0, T] ; \Bbb {R}^d)$, stands for the 
first exit time from $D$, that is
\begin{eqnarray*}
\tau _D(w) = \inf \{ t\in [0, T] ; w(t)\notin D \} . 
\end{eqnarray*}

Let us define a second order differential operator $\mathscr {L}^{\varepsilon }$ by 
\begin{eqnarray*}
\mathscr {L}^\varepsilon = 
\frac{1}{2}\sum ^d_{i, j = 1}a^{ij}(x, \varepsilon )\frac{\partial ^2}{\partial x^i\partial x^j} + 
\sum ^d_{i = 1}b^i(x, \varepsilon )\frac{\partial }{\partial x^i} - c(x, \varepsilon ), 
\end{eqnarray*}
where $a^{ij} = \sum ^d_{k = 1}\sigma ^{ik}\sigma ^{jk}$. 
We consider the following Cauchy-Dirichlet problem for a PDE of parabolic type:
\begin{eqnarray}\label{eq_PDE}
\left\{
\begin{array}{ll}
 	\frac{\partial }{\partial t}u^\varepsilon (t, x) + \mathscr {L}^\varepsilon u^\varepsilon (t, x) = 0,& (t, x)\in [0, T)\times D,	\vspace{1mm}\\
 	u^\varepsilon (T, x) = f(x),& x\in D,	\\
 	u^\varepsilon (t, x) = 0,& (t, x)\in [0, T]\times \partial D. 
\end{array}
\right.
\end{eqnarray}

Now we introduce a series of the assumptions necessary for the existence and the uniqueness
of the classical solution of (\ref{eq_PDE}).
\begin{description}
 \item[\mbox{[A]}] \ There is a positive constant $A_1$ such that 
\begin{eqnarray*}
|\sigma ^{ij}(x, \varepsilon )|^2 + |b^i(x, \varepsilon )|^2 \leq A_1(1 + |x|^2), \ \ x \in \Bbb{R}^d, \ \varepsilon \in I, \ i, j = 1, \ldots , d. 
\end{eqnarray*}
Moreover, for each $\varepsilon \in I$ it holds that $\sigma ^{ij}(\cdot , \varepsilon ), b^i(\cdot , \varepsilon )\in \mathcal {L}$ 
for $i, j = 1, \ldots , d$, where 
$\mathcal {L}$ is the set of locally Lipschitz continuous functions defined on $\Bbb{R}^d$:
\beas
\mathcal {L} &=& \{f\in C(\Bbb{R}^d; \Bbb{R});\ 
\mbox{for any compact set $K\subset \Bbb{R}^d$},\\
&&
\ 
\exists C_K>0\ \mbox{such that}\ |f(x)-f(y)| \leq C_K |x-y|, x, y \in K \}
\eeas
\end{description}
\begin{rem}
Note that under [A], the existence and uniqueness of a solution of (\ref {eq_SDE_epsilon}) are guaranteed 
on any filtered probability space equipped with a standard $d$-dimensional Brownian motion, and 
Corollary 2.5.12 in \cite {Krylov} and Lemma 3.2.6 in \cite {Nagai} imply 
\begin{eqnarray}\label{SDE_poly_growth}
\E [\sup _{0\leq r\leq t}|X^{\varepsilon , x}_r - x|^{2l}] \leq C_{l}t^{l - 1}(1 + |x|^{2l}), \ \ (t, x)\in [0, T]\times \Bbb {R}^d, \ l\in \Bbb {N}
\end{eqnarray}
for some $C_{l} > 0$ which depends only on $l$ and $A_1$. 
Moreover, 
$(X^x_r)_r$ has the strong Markov property. 
\end{rem}
\begin{description}
 \item[\mbox{[B]}] \ The function $f(x)$ is continuous on $\bar{D}$ and 
there are $C_f > 0$ and $m\in \Bbb {N}$ such that $|f(x)| \leq C_f(1 + |x|^{2m})$, $x\in \Bbb {R}^d$. 
Moreover, $f(x) = 0$ on $\Bbb {R}^d\setminus D$. 
\end{description}
\begin{rem}
The assumption $[B]$ guarantees the continuity of a solution of  (\ref{eq_PDE})  on the so called parabolic boundary 
$\Sigma = \partial D \times [0,T) \cup \bar{D} \times \{T\}$, 
in addition to the continuity and polynomial growth of $f$.
\end{rem}

\begin{description}
 \item[\mbox{[C]}] \ 
$c(x, \varepsilon )$ is non-negative (i.e. $c(x, \varepsilon )\geq 0$). 
Moreover, for each $\varepsilon \in I$, 
it holds that $c(\cdot , \varepsilon )\in \mathcal {L}$. 
 \item[\mbox{[D]}] \ 
The boundary $\partial D$ has the outside strong sphere property, that is, 
for each $x\in \partial D$ there is a closed ball $E$ such that $E\cap D = \phi $ and $E\cap \bar{D} = \{x\}$. 
\end{description}
\begin{rem}
The assumption $[D]$ provides the regularity of each point in $\partial D$.
($c.f. \cite{Friedman1}$)
Also, \cite{Rubio} points out that $[D]$ with the ellipticity of the matrix $(a^{ij}(x, \varepsilon ))_{ij}$
in $[E]$ below gives
\beas
P( \tau_{D}(X^{\ep,x}) = \tau_{\bar{D}}(X^{\ep,x})) =1.
\eeas
\end{rem}
\begin{description}
 \item[\mbox{[E]}] \ 
The matrix $(a^{ij}(x, \varepsilon ))_{ij}$ is locally elliptic in the sense that 
for each $\varepsilon \in  I $ and compact set $K\subset \Bbb {R}^d$ 
there is a positive number $\mu _{\varepsilon , K}$ such that 
$\sum ^d_{i, j = 1}a^{ij}(x, \varepsilon )\xi ^i\xi ^j \geq \mu _{\varepsilon , K}|\xi |^2 $ 
for any 
$x\in K$ and $\xi \in \Bbb {R}^d$.
\end{description}
\begin{rem}
Note that although the condition [E] (local ellipticity) is necessary for the existence of classical solution of our PDE (See Remark 2.2 in \cite{Rubio}),
the assumption 
can be removed through consideration of viscosity solutions rather than classical solutions 
by applying Theorem 8.2 in \cite{Crandall-Lions-Ishii} and Theorems 4.4.3 and 7.7.2 in \cite {Nagai}. 
Note that we need the additional assumption such that $I\subset [0, \infty )$ by technical reason in this case.
\end{rem}

Under the assumptions [A]-[E] above, we have the following existence and uniqueness result
due to  
Theorem 3.1 in \cite {Rubio}. 
\begin{thm} \ \label{th_viscosity}Assume $[A]$--$[E]$. 
For each $\varepsilon \in  I $, $u^\varepsilon (t, x)$ is a $($classical$)$ solution of $(\ref {eq_PDE})$ and 
\begin{eqnarray}\label{ineq_poly}
\sup _{(t, y)\in [0, T]\times \bar{D}}|u^\varepsilon (t, x)| / (1 + |x|^{2m}) < \infty . 
\end{eqnarray}
Moreover, if $w^\varepsilon (t, x)$ is also a solution of $(\ref {eq_PDE})$ 
satisfying the growth condition 
\begin{eqnarray*}
\sup _{(t, y)\in [0, T]\times \bar{D}}|w^\varepsilon (t, x)| / (1 + |x|^{2m'}) < \infty  
\end{eqnarray*}
for some $m'\in \Bbb {N}$, 
then $u^\varepsilon = w^\varepsilon $. 
\end{thm}

\section{Asymptotic Expansion of Barrier Option Price}\label{AE}
Our purpose is to present an asymptotic expansion of the barrier option price $u^\varepsilon (t, x)$:
\begin{eqnarray}
\label{eq_ae1}
u^\varepsilon (t, x) = u^0(t, x) + \varepsilon v^0_1(t, x) + \cdots  + \varepsilon ^{n-1}v^0_{n-1}(t, x) + O(\varepsilon ^n), \ \ \varepsilon \rightarrow 0.
\end{eqnarray}
Here, the coefficient function $v^0_k(t, x)$, $k = 1, \ldots , n - 1$ are (formally) given as the solution of 
\begin{eqnarray}\label{eq_PDE2}
\left\{
\begin{array}{ll}
 	\frac{\partial }{\partial t}v^0_k(t, x) + \mathscr {L}^0v^0_k(t, x) + g^0_k(t, x) = 0,& (t, x) \in [0, T)\times D,	\vspace{1mm}\\
 	v^0_k(T, x) = 0,& x\in D,	\\
 	v^0_k(t, x) = 0,& (t, x)\in [0, T]\times \partial D, 
\end{array}
\right.
\end{eqnarray}
where $g^0_k(t, x)$ is given inductively by 
\begin{eqnarray}\label{def_g}
g^0_k(t, x) = \tilde{\mathscr {L}}^0_ku^0(t, x) + 
\sum ^{k-1}_{l = 1}\tilde{\mathscr {L}}^{0}_{k-l}v^{0}_l(t, x),
\end{eqnarray}
where $\tilde{\mathscr {L}}^0_k$ is defined as follows:
\begin{eqnarray}\label{g_expansion}
\tilde{\mathscr {L}}^0_k = 
\frac{1}{k!}\left\{ 
\frac{1}{2}\sum ^d_{i, j = 1}\frac{\partial ^ka^{ij}}{\partial \varepsilon ^k}(x, 0)\frac{\partial ^2}{\partial x^i\partial x^j} + 
\sum ^d_{i = 1}\frac{\partial ^kb^i}{\partial \varepsilon ^k}(x, 0)\frac{\partial }{\partial x^i} - 
\frac{\partial ^kc}{\partial \varepsilon ^k}(x, 0) \right\}, 
\end{eqnarray}

To study the asymptotic expansion,
we put the following assumptions in addition to [A]--[E].
Firstly, by the next condition we can properly define $\tilde{\mathscr {L}}^0_k$, $k\in \Bbb {N}$ in (\ref{g_expansion}) above.
\begin{description}
 \item[\mbox{[F]}] \ Let $n\in \Bbb {N}$. The functions $a^{ij}(x, \varepsilon )$, 
$b^i(x, \varepsilon )$ and $c(x, \varepsilon )$ 
are $n$-times continuously differentiable in $\varepsilon $. 
Furthermore, each of derivatives 
$\partial ^ka^{ij}/\partial \varepsilon ^k$, $\partial ^kb^i/\partial \varepsilon ^k$, 
$\partial ^kc/\partial \varepsilon ^k$, 
$k = 1, \ldots , n - 1$, has a polynomial growth rate in $x\in \Bbb {R}^d$ 
uniformly in 
$\varepsilon \in  I $. 
\end{description}

To state the existence of the functions $v^0_k(t, x)$ in the asymptotic expansion (\ref{eq_ae1}), 
we first prepare the following set. 
\begin{defn}\label{def1}
The set $\mathcal {H}^{m, p}$ of $g\in C([0, T)\times \bar{D})$ 
is defined to satisfy the following condition: 
There is some $M^g\in C([0, T))\cap L^p([0, T), dt)$ such that 
\begin{eqnarray}\label{def_H_map_growth}
|g(t, x)| \leq M^g(t)(1 + |x|^{2m}), \ \ t\in [0, T), \ x, y\in \bar{D}.
\end{eqnarray}
\end{defn}
Given this definition of the set $\mathcal {H}^{m, p}$, we put the next condition on $u^0$. 
\begin{description}
 \item[\mbox{[G]}] \ $u^0\in \mathcal {G}^{m}$, where 
\begin{eqnarray*}
\mathcal {G}^{m} &=& 
\Big \{ g\in C^{1, 2}([0, T)\times D)\cap C([0, T]\times \bar{D})\ ; \\&&\hspace{40mm}
\frac{\partial g}{\partial x^i}\in \mathcal {H}^{m, 2}, \ 
\frac{\partial ^2g}{\partial x^i\partial x^j}\in \mathcal {H}^{m, 1}, \ 
i, j = 1, \ldots , d\Big \}. 
\end{eqnarray*}
\end{description}

Now we examine the conditions necessary for the classical solution
to the PDE (\ref {eq_PDE2}). 
Let us start with the case of $k = 1$. 
By the assumption [G], we have
$g^0_1\in \mathcal {H}^{m, 1}$ 
for some $m\in \Bbb {N}$ by the definition of $g^0_k$ with $k = 1$ in (\ref{def_g}).
Thus we can define 
\begin{eqnarray}\label{form_FK_v1}
v^0_1(t, x) = \E \left[ \int ^{(T-t)\wedge \tau _D(X^{0, x})}_0
\exp \left( -\int ^r_0c(X^{0, x}_v, 0)dv
\right) g^0_1(t + r, X^{0, x}_r)dr\right ] . 
\end{eqnarray}
Therefore, if we assume that $v^0_1\in C^{1, 2}([0, T)\times D)$ 
we can show that $v^0_1$ is the solution of (\ref {eq_PDE2}) with $k=1$, 
that is, we can confirm that 
\begin{eqnarray*}
\frac{\partial }{\partial t}v^0_1(t, x) + \mathscr {L}^0v^0_1(t, x) + g^0_1(t, x) = 0. 
\end{eqnarray*}
Note that the relations $v^0_1(T, \cdot ) = 0$ and $v^0_1 = 0$ on $[0, T]\times \partial D$ are obvious. 

Next, let us give some comments on the smoothness of $v^0_1$. 
In many cases as in the Black--Scholes model
(see (\ref{density})
in Section 4).
we can rewrite (\ref{form_FK_v1})
as 
\begin{eqnarray*}
v^0_1(t, x) = \int ^{T-t}_0\int _Dg^0_1(t+r, y)p(r, x, y)dydr 
\end{eqnarray*}
for some $p(r, x, y)$. 
Thus, if $p$ has a ``good'' smoothness property, the smoothness of $v^0_1$ also holds such as 
\begin{eqnarray*}
\frac{\partial }{\partial t}v^0_1(t, x) &=& - \lim _{s\rightarrow T}\int _Dg^0_1(s, y)p(s-t, x, y)dy + 
\int ^t_0\int _D\frac{\partial }{\partial t}g^0_1(t+r, y)p(r, x, y)dr, \\
\frac{\partial }{\partial x^i}v^0_1(t, x) &=& 
\int ^t_0\int _Dg^0_1(t+r, y)\frac{\partial }{\partial x^i}p(r, x, y)dr, \\
\frac{\partial ^2}{\partial x^i\partial x^j}v^0_1(t, x) &=& 
\int ^t_0\int _Dg^0_1(t+r, y)\frac{\partial ^2}{\partial x^i\partial x^j}p(r, x, y)dr. 
\end{eqnarray*}
Moreover, if $v^0_1$ is in $\mathcal {G}^{m_1}$ for some $m_1\in \Bbb {N}$,
we also have $g^0_2\in \mathcal {H}^{\tilde{m}_1, 1}$ for some $\tilde{m}_1\in \Bbb {N}$ 
by the definition of $g^0_k$ with $k = 2$ in (\ref{def_g}). Then, we can define $v^0_2$ similarly as $v^0_1$. 
Furthermore, under some suitable smoothness conditions for $v^0_2$, which may be given by
the smoothness property of $p(r,x,y)$,
we are able to show that $v^0_2$ is the classical solution of (\ref {eq_PDE2}) with $k = 2$. 

Thus, the observation above leads us to our final assumption.
\begin{description}
 \item[\mbox{[H]}] \ It holds that $v^0_k\in \mathcal {G}^{m_n}$, $k = 1, \ldots , n - 1$ for some $m_n\in \Bbb {N}$, where 
\begin{eqnarray}\label{form_FK}
v^0_k(t, x) = \E \left[ \int ^{(T-t)\wedge \tau _D(X^{0, x})}_0
\exp \left( -\int ^r_0c(X^{0, x}_v, 0)dv
\right) g^0_k(t + r, X^{0, x}_r)dr\right ] . 
\end{eqnarray}
\end{description}
Then, we can show the next result. The proof is given in Section \ref{sec_proof_th2} of Appendix. 
\begin{thm}\label{th_classical_vk}
Assume $[A]$--$[H]$. 
Then, for each $k = 1, \ldots , n-1$,  $v^0_k$ is the classical solution of $(\ref {eq_PDE2})$ 
and satisfies 
\begin{eqnarray}\label{est_v0k}
|v^0_k(t, x)| \leq C_k(1 + |x|^{2m_k}), \ \ (t, x)\in [0, T]\times \Bbb {R}^d
\end{eqnarray}
for some $C_k, m_k > 0$. 
\end{thm}

Note that the uniqueness of the solutions of (\ref {eq_PDE2}) follows from the same arguments 
as in the proof of 
Theorem 5.7.6 in \cite {Karatzas-Shreve}. That is, we obtain the next proposition. 
\begin{prop} \ \label{prop_uniqueness_v2}
For any function $g$ which has a polynomial growth rate in $x$ uniformly in $t$, 
a classical solution of $(\ref {eq_PDE2})$ is unique in the following sense: \ 
if $v$ and $w$ are classical solutions of $(\ref {eq_PDE2})$ and 
$|v(t, x)| + |w(t, x)| \leq C(1 + |x|^{2m})$ for some $C, m > 0$, then $v = w$. 
\end{prop}

Now, we are able to state our first main result on the asymptotic expansion. The proof is given in Section \ref{sec_proof}
of Appendix.
\begin{thm} \ \label{th_main} 
Assume $[A]$--$[H]$. 
There are positive constants $C_n$ and $\tilde{m}_n$ which are independent of $\varepsilon $ such that 
\begin{eqnarray*}
\left |u^\varepsilon (t, x) - (u^0(t, x) + \sum ^{n-1}_{k = 1}\varepsilon ^kv^0_k(t, x))\right | \leq C_n(1 + |x|^{2\tilde{m}_n})\varepsilon ^n, \ \ 
(t, x)\in [0, T]\times \bar{D}. 
\end{eqnarray*}
\end{thm}

Next,
we construct a semi-group corresponding to $(X^{0, x}_t)_t$
(that is, $(X^{\varepsilon, x}_t)_t$ with $\varepsilon=0$) and $D$.
Then, based on this semi-group
we can obtain more explicit representation for the coefficient function $v^0_k(t, x)$ 
than the right hand side of (\ref {form_FK}) .

Let $C^0_b(\bar{D})$ be the set of bounded continuous functions $f : \bar{D}\longrightarrow \Bbb {R}$ 
such that $f(x) = 0$ on $\partial D$. 
Obviously, $C^0_b(\bar{D})$ equipped with the sup-norm becomes a Banach space.

For $t\in [0, T]$ and $f\in C^0_b(\bar{D})$, we define 
$P^D_tf : \bar{D}\longrightarrow \Bbb{R}$ by 
\begin{eqnarray}\label{semigroup_rep}
P^D_tf(x) = \E \left[
\exp \left( -\int ^t_0c(X^{0, x}_v, 0)dv\right)
f(X^{0, x}_t)1_{\{\tau _D(X^{0, x})\geq t\}}\right],
\end{eqnarray}
where 
$c(x, 0)$ is non-negative.
We notice that $P^D_tf(x)$ is equal to $u^0(T - t, x)$ with the payoff function $f$. 
Then, we have the following result:
\begin{prop}
\ \label{th_semigroup}
Under the assumptions [A]--[E], 
the mapping $P^D_t : C^0_b(\bar{D})\longrightarrow C^0_b(\bar{D})$ is well-defined 
and $(P^D_t)_{0\leq t\leq T}$ is a contraction semi-group. 
\end{prop}

\begin{proof} 
Let $f\in C^0_b(\bar{D})$. 
The relations $P^D_0f = f$, $P^D_tf|_{\partial D} = 0$ and 
$\sup _{\bar{D}}|P^D_tf| \leq \sup _{\bar{D}}|f|$ are obvious. 
The continuity of $P^D_tf$ is by Lemma 4.3 in \cite {Rubio}. 
The semi-group property is verified by a straightforward calculation. 
\end{proof}
\begin{rem}
Note that $(P^D_t)_t$ also has the semi-group property on the set $C^0_p(\bar{D})$ of 
continuous functions $f$, each of which has a polynomial growth rate and satisfies $f(x) = 0$ on $\partial D$. 
\end{rem}

Finally, we show our second main result
on the semi-group representation of the coefficient function $v^0_k$ in the expansion, 
whose proof is given in Section \ref{proof_form_semigroup} in Appendix.
\begin{thm} \ \label{th_form_semigroup}
Under Assumptions [A]--[H], 
for each $k = 1, \ldots , n - 1$ 
\begin{eqnarray}
&&v^0_k(T - t, x)\nonumber\\
&& = \sum ^k_{l = 1}\sum _{(\beta ^i)^l_{i=1}\subset \Bbb {N}^l, \sum _i\beta ^i = k}
\int ^t_0\int ^{t_1}_0\cdots \int ^{t_{l - 1}}_0
P^D_{t - t_1}\tilde{\mathscr {L}}^0_{\beta ^1}P^D_{t_1 - t_2}\tilde{\mathscr {L}}^0_{\beta ^2}\cdots P^D_{t_{l - 1} - t_l}\tilde{\mathscr {L}}^0_{\beta ^l}P^D_{t_l}f(x)dt_l\cdots dt_1.\nonumber\\
\label{semigroup_formula} 
\end{eqnarray}
\end{thm}

\section{Application to Barrier Option Pricing in Stochastic Volatility
Environment}\label{sec_examples}
This section demonstrates the effectiveness of our method
in stochastic volatility environment: Section 4.1 derives concrete approximation formulas,
and Section 4.2 shows numerical examples.
\subsection{Approximation 
of Barrier Option Prices in a Stochastic Volatility Model}
We consider the following stochastic volatility model.
\begin{eqnarray}\label{SVdrift}
dS_{t}^{\varepsilon}&=&(c-q) S_{t}^{\varepsilon} dt+\sigma_t^{\varepsilon} S_{t}^{\varepsilon} dB_{t}^1, \ S_0^{\varepsilon}=S,\\
d\sigma_{t}^{\varepsilon}&=&\varepsilon \lambda (\theta-\sigma_t^{\varepsilon})dt
+\varepsilon \nu \sigma_{t}^{\varepsilon}(\rho dB_{t}^1+\sqrt{1-\rho^2} dB_{t}^2), \ \sigma_{0}^{\varepsilon}=\sigma,\nn 
\end{eqnarray}
where $c, q >0$, $\varepsilon \in [0,1)$, $\lambda, \theta, \nu>0$, $\rho \in [-1,1]$ and $B=(B^1,B^2)$ is a two dimensional Brownian motion. 
Here $c$ and $q$ represent a domestic interest rate and a foreign interest rate, respectively when we consider the currency options. 
Clearly, applying It\^o's formula, we have its logarithmic process:
\begin{eqnarray}\label{log_process}
dX_{t}^{\varepsilon}&=&(c-q -\frac{1}{2}(\sigma_t^{\varepsilon})^2) dt+\sigma_t^{\varepsilon} dB_{t}^1, \ X_{0}^{\varepsilon}=x=\log S, \\
d\sigma_{t}^{\varepsilon}&=&\varepsilon \lambda (\theta-\sigma_t^{\varepsilon})dt+\varepsilon \nu\sigma_{t}^{\varepsilon}(\rho dB^1_t + \sqrt{1 - \rho ^2}dB^2_t), \ \sigma_{0}^{\varepsilon}=\sigma.\nn
\end{eqnarray}
Also, its generator is expressed as
\begin{eqnarray}
\mathscr {L}^{\varepsilon}=\left(c-q -\frac{1}{2}\sigma^2 \right)\frac{\partial}{\partial x}+\frac{1}{2}\sigma^2\frac{\partial^2}{\partial x^2}
+\varepsilon \rho \nu\sigma^2 \frac{\partial^2}{\partial x \partial \sigma}+\varepsilon \lambda (\theta-\sigma)\frac{\partial}{\partial \sigma}
+\varepsilon^2 \frac{1}{2}\nu^2\sigma^2 \frac{\partial^2}{\partial  \sigma^2}.\label{Gene_ep}
\end{eqnarray}
In this case, $\tilde{\mathscr {L}}_1^0$ which is defined by (\ref{g_expansion}) with $k=1$ is given as 
\begin{eqnarray}
\tilde{\mathscr {L}}_1^0 &=&\rho \sigma^2{\partial^2\over \partial x \partial \sigma}+\lambda (\theta-\sigma)\frac{\partial}{\partial \sigma}. \label{Gene_1}
\end{eqnarray}
We will apply the asymptotic expansion in the previous section
to (\ref{log_process})
and give an approximation formula for 
a barrier option price, which is given under a risk-neutral probability measure as
\begin{eqnarray*}
C_\mathrm {Barrier}^{SV,\varepsilon}(T-t,e^x)
= \E \left[e^{-c(T-t)} {f}(S^{\varepsilon,e^x} _{T-t} ) 1_{\{\tau _{(L,\infty)}(S^{\varepsilon, e^x}) > T-t\}}\right],
\end{eqnarray*}
where $f$ stands for a 
payoff function  and $L(< S)$ is a barrier price.

Then, $u^{\varepsilon}(t,x)=C_\mathrm {Barrier}^{SV,\varepsilon}(T-t,e^x)$ satisfies the following PDE:
\begin{eqnarray}\label{SV_PDE_drift}
\left\{ 
\begin{array}{ll}
\left(\frac{\partial}{\partial t}+\mathscr {L}^{\varepsilon}-c \right)u^{\varepsilon}(t,x) = 0, & (t, x)\in (0, T]\times D, \\
u^{\varepsilon}(T,x) = \bar{f}(x), & x \in \bar{D}, \\
u^{\varepsilon}(t,l) = 0, & t\in [0, T]. 
\end{array}
\right. 
\end{eqnarray}
where $\bar{f}(x)=\max \{ e^x - K, 0 \}$, $D = (l, \infty )$ and $l=\log L$.
We obtain the $0$-th order  $u^0$ as
\begin{eqnarray}
u^{0}(t,x)=P^{D}_{T-t}{\bar f}(x)=\E[e^{-c(T-t)}{\bar f}(X_{T-t}^{x,0})1_{\{\tau _{D}(X^{0, x}) > T-t\}} ]. 
\end{eqnarray}
\begin{rem}
$u^0$ satisfies the PDE (\ref{SV_PDE_drift}) with $\varepsilon = 0$. 
Although the condition [E] in Section 2 does not seem to be satisfied in this case, 
the volatility process $(\sigma_{t}^{\varepsilon})_t$ becomes a constant $\sigma > 0$, and
so (\ref{log_process}) is reduced to a one-dimensional SDE. 
Then, (\ref{SV_PDE_drift}) with $\varepsilon = 0$ becomes a non-degenerating PDE with fixed $\sigma$. 
Therefore, we need not take care of the lack of the condition [E] in this example. 
\end{rem}

Setting $\alpha=c-q$, we note that 
${P}^{D}_{T-t}{\bar f}(x)=C^{BS}_\mathrm{Barrier}(T-t ,e^x,\alpha,\sigma,L)$ 
is the price of the down-and-out barrier call option under the Black-Scholes model:
\begin{eqnarray}
C^{BS}_\mathrm{Barrier}(T-t ,e^x,\alpha,\sigma,L)=C^{BS}\l(T-t , {e^x},\alpha,\sigma \r)-
\left( \frac{e^x}{L} \right)^{1-\frac{2\alpha }{\sigma^2}} 
C^{BS}\l(T-t,\frac{L^2}{e^x},\alpha,\sigma \r). 
\label{BS_barrier_drift}
\end{eqnarray}
Here, we recall that the price of the plain vanilla option under the Black-Scholes model is given as
\begin{eqnarray}
C^{BS}\l(T-t, {e^x},\alpha,\sigma \r)=e^{-q (T-t)}e^x N(d_{1}(T-t,x,\alpha))-e^{-c(T-t)}KN(d_2(T-t,x,\alpha)),
\end{eqnarray}
where
\begin{eqnarray*}
d_{1}(t,x,\alpha)&=&\frac{x-\log K+\alpha t}{\sigma \sqrt{t}}+\frac{1}{2}\sigma \sqrt{t}, \\
d_{2}(t,x,\alpha)&=&d_{1}(t,x,\alpha)-\sigma \sqrt{t}\\
N(x)&=&\int_{-\infty}^{x}n(y)dy,\\
n(y)&=&\frac{1}{\sqrt{2 \pi}}e^{\frac{-y^2}{2}}.
\end{eqnarray*}

Note also that
\begin{eqnarray*}
P(\tau _D(X^{0, x}) \geq t | X^{0, x}_t) = 
1 - \exp \left( -\frac{2(x - l)(X^{0, x}_t - l)}{\sigma ^2t}\right) \ \ 
\mathrm {on} \ \{ X^{0, x}_t > l \} . 
\end{eqnarray*}
Therefore, for $g\in C^0_p(\bar{D})$ we have 
\begin{eqnarray}\label{eq_semi1}
P^D_tg(x) = \E [P(\tau _D(X^{0, x}) \geq t | X^x_t)e^{-ct} g(X^{0, x}_t)1_{\{ X^{0, x}_t > l\} }] = 
\int ^\infty _l e^{-ct} g(y)p(t, x, y)dy, 
\end{eqnarray}
where 
\begin{eqnarray}\label{density}
p(t, x, y) &=& \frac{1}{\sqrt{2\pi \sigma ^2t}}
(1 - e^{-\frac{2(x-l)(y-l)}{\sigma ^2t}})e^{-\frac{(y - x - \mu t)^2}{2\sigma ^2t}},\\
\mu &=& \alpha -\sigma^2/2 = (c-q-\sigma^2/2). \nn
\end{eqnarray}

Then, we show the following main result in this section.
\begin{thm} \ \label{DOBarrier2}
We obtain an approximation formula for the down-and-out barrier call option under the stochastic volatility model (\ref{SVdrift}): 
\begin{eqnarray}
C_\mathrm {Barrier}^{SV,\varepsilon}(T,e^x) = C^{BS}_\mathrm{Barrier}(T ,e^x,\alpha,\sigma,L)
+\varepsilon v^0_1(0,x)+O(\varepsilon^2), \label{Barrier_drift_approx}
\end{eqnarray}
where
\begin{eqnarray}\label{1st_order}
v^0_1(0,x) = e^{-c T}\int_{0}^{T} \int_{l}^{\infty}\frac{1}{\sqrt{2 \pi \sigma^2 s}}
(1-e^{-\frac{2(x-l)(y-l)}{\sigma^2s}})e^{-\frac{(y-x-(\alpha-\frac{1}{2}\sigma^2) s )^2}{2\sigma^2 s}} \vartheta(s,y)dyds, 
\end{eqnarray}
\begin{eqnarray}
&&\vartheta(t,x)\nn\\
&=&e^{\alpha(T-t)}\rho \nu\sigma e^x n(d_{1}(T-t,x,\alpha))(-d_{2}(T-t,x,\alpha))\nn\\
&+&2 e^{\alpha(T-t)}\rho \nu   {\alpha}
\left( \frac{e^x}{L} \right)^{-\frac{2\alpha}{\sigma^2}} Ln(c_{1}(T-t,x,\alpha))\sqrt{T-t}\nn\\
&-&e^{\alpha(T-t)}\rho \nu\sigma \left( \frac{e^x}{L} \right)^{-\frac{2\alpha}{\sigma^2}} L n(c_{1}(T-t,x,\alpha))c_1(T-t,x,\alpha)\nn\\
&-&e^{c(T-t)}\rho \nu  \frac{4\alpha}{\sigma} \left( \frac{e^x}{L} \right)^{1-\frac{2\alpha}{\sigma^2}}\nn\\
&\times& \Biggl\{ C^{BS}\l(T-t,\frac{L^2}{e^x},\alpha,\sigma \r)\l\{ 1+(x-\log{L}) 
\l(1-\frac{2\alpha}{\sigma^2} \r)  \r\}\nn\\
&& -(x-\log L) e^{-q (T-t)} \frac{L^2}{e^x} N(c_1(T-t,x,\alpha)) \Biggr\}\nn\\
&+&\lambda (\theta-\sigma) e^{\alpha(T-t)}e^x  n(d_{1}(T-t,x,\alpha))\sqrt{T-t} \nn\\
&-&\lambda (\theta-\sigma) \left( \frac{e^x}{L} \right)^{-\frac{2 \alpha}{\sigma^2}} e^{\alpha (T-t)}L n(c_{1}(T-t,x,\alpha))\sqrt{T-t}\nn\\
&-&e^{c(T-t)}\lambda (\theta-\sigma) \frac{4\alpha}{\sigma^3}\l(\log\frac{e^x}{L}\r)\left( \frac{e^x}{L} \right)^{1-\frac{2\alpha}{\sigma^2}}C^{BS}\l(T-t,\frac{L^2}{e^x},\alpha,\sigma \r), \label{Greeks}
\end{eqnarray}
and
\begin{eqnarray*}
c_1(t,x, \alpha)&=& {2l-x-\log K +\alpha t \over \sigma\sqrt{t}} +\frac{1}{2}\sigma \sqrt{t}.
\end{eqnarray*}
\end{thm}
\begin{proof}
Firstly, note that when $k=1$ in Theorem \ref{th_form_semigroup}, we have
\begin{eqnarray*}
v^0_1(T - t, x) &=&
 \int ^t_0P^D_{t - r}\tilde{\mathscr {L}}^0_1P^D_rf(x)dr. 
\end{eqnarray*}
Thus, 
we see the expansion 
\begin{eqnarray}
C_\mathrm {Barrier}^{SV,\varepsilon}(T-t,e^x) = C^{BS}_\mathrm{Barrier}(T-t ,e^x,\alpha,\sigma,L)
+\varepsilon \int_{0}^{T-t}P^{D}_{s}\tilde{\mathscr {L}}_1^0P^{D}_{T-t-s}{\bar f}(x)ds+O(\varepsilon^2).
\end{eqnarray}
The first-order approximation term $v^0_1(t,x)=\int_{0}^{T-t}P^{D}_{s}\tilde{\mathscr {L}}_1^0P^{D}_{T-t-s}{\bar f}(x)ds$ is given by 
\begin{eqnarray*}
v^0_1(t,x)
&=&\int_{0}^{T-t} e^{-cs}{\bar P}^{D}_{s}\tilde{\mathscr {L}}_1^0 e^{-c(T-t-s)}{\bar P}^{D}_{T-t-s}{\bar f}(x)ds\\
&=&e^{-c(T-t)}\int_{0}^{T-t}{\bar P}^{D}_{s} \tilde{\mathscr {L}}_1^0 {\bar P}^{D}_{T-t-s} {\bar f}(x)ds,
\end{eqnarray*}
where ${\bar P}_t^{D}$ is defined by 
\begin{eqnarray*}
{\bar P}_t^{D}\bar{f}(x)=\int_{l}^{\infty}\frac{1}{\sqrt{2 \pi \sigma^2 s}}
(1-e^{-\frac{2(x-l)(y-l)}{\sigma^2s}})e^{-\frac{(y-x-(\alpha-\frac{1}{2}\sigma^2) s )^2}{2\sigma^2 s}} \bar{f}(y)dy.
\end{eqnarray*}
Define $\vartheta(t,x)$ as 
\begin{eqnarray*}
\vartheta(t,x)&=&\tilde{\mathscr {L}}_1^0 {\bar P}_{T-t}^{D}f(e^x)\\
&=&e^{c(T-t)} \rho \nu\sigma^2 \frac{\partial^2}{\partial x \partial \sigma} C^{BS}_\mathrm{Barrier}(T-t ,e^x,\alpha,\sigma,L)+e^{c(T-t)}\lambda (\theta-\sigma) \frac{\partial}{\partial \sigma} C^{BS}_\mathrm{Barrier}(T-t ,e^x,\alpha,\sigma,L).
\end{eqnarray*}
A straightforward calculation shows that the above function 
agrees with the right-hand side of (\ref{Greeks}).
Then we get the assertion. 
\end{proof} 

Remark that through numerical integrations with respect to time $s$ and space $y$ in 
(\ref{1st_order}),
we easily obtain the first order approximation of the down-and-out option prices.\\

Next, as a special case of (\ref{SVdrift}) we consider the following stochastic volatility model with no drifts:
\begin{eqnarray}\label{no_drift}
dS^{\varepsilon }_t &=&\sigma ^\varepsilon _t S_{t}^{\varepsilon}dB^1_t, \ \  S^{\varepsilon }_0 = S>0, \\
d\sigma ^\varepsilon _t &=& \varepsilon \nu \sigma ^\varepsilon _t(\rho dB^1_t + \sqrt{1 - \rho ^2}dB^2_t), \ \ \sigma ^\varepsilon _0 = \sigma > 0 .\nn
\end{eqnarray}
where $\varepsilon \in [0,1)$, $\rho \in [-1,1]$ and $B=(B^1,B^2)$ is a two dimensional Brownian motion. 
In this case, we can provide a simpler approximation formula
than in Theorem \ref{DOBarrier2}.  

By It\^o's formula, the following logarithmic model is obtained. 
\begin{eqnarray}\label{SDE_logSABR}
\begin{split}
dX^{\varepsilon }_t &= -\frac{1}{2}(\sigma ^\varepsilon _t)^2dt+  \sigma ^\varepsilon _tdB^1_t, \ \ X^{\varepsilon }_0 = x=\log S, \\
d\sigma ^\varepsilon _t &= \varepsilon \nu \sigma ^\varepsilon _t(\rho dB^1_t + \sqrt{1 - \rho ^2}dB^2_t), \ \ \sigma ^\varepsilon _0 = \sigma.
\end{split}
\end{eqnarray}
Again, the barrier option price is given by
\begin{eqnarray*}
C_\mathrm {Barrier}^{SV,\varepsilon}(T,e^x)
= \E \left[ f(S^\varepsilon _T ) 1_{\{\min _{0\leq u\leq T}S^\varepsilon _u > L\}}\right],
\end{eqnarray*}
where $f$ stands for a 
payoff function  and $L(< S)$ is a barrier price.

The differentiation operators $\mathscr {L}^{\varepsilon}$, $\tilde{\mathscr {L}}_1^0$ and 
the PDE are same as 
(\ref{Gene_ep})--(\ref{SV_PDE_drift})
with $c=q=0$ and $\lambda=0$. 
Also, the barrier option price in the Black-Scholes model 
coincides with 
(\ref{BS_barrier_drift}) with no drift, that is,  
\begin{eqnarray*}
C_\mathrm {Barrier}^{BS}(T,S)=C^{BS}(T,S)-\l(\frac{S}{L} \r)C^{BS}\l(T,\frac{L^2}{S} \r),
\end{eqnarray*}
where $C^{BS}(T,S)$ is the driftless Black-Scholes formula of the European call option given by
\begin{eqnarray*}
C^{BS}(T,S)=S N(d_1(T,\log S))-K N(d_2(T,\log S))
\end{eqnarray*}
with
\begin{eqnarray*}
d_1(t,x)&=&d_1(t,x,0)= \frac{ x-\log K+\sigma^2t/2}{ \sigma \sqrt{t}},\\
d_2(t,x)&=&d_2(t,x,0)=d_1(t,x)-\sigma \sqrt{t}.
\end{eqnarray*}
Then, we reach the following expansion formula which only needs 1-dimensional numerical integration.
\begin{thm} \ \label{DOBarrier}
$C_\mathrm {Barrier}^{SV,\varepsilon}(T,e^x)
=C_\mathrm {Barrier}^{BS}(T,e^x)+\varepsilon v_1^0(0,x)+O(\varepsilon^2)$, 
where
\begin{eqnarray}
v_1^0(0,x) &=& -{1\over 2}T \nu \rho  \sigma \l\{e^x n(d_1(T,x))d_2(T,x)+Ln(c_1(T,x))c_1(T,x)\r\}\nn\\
&&+\frac{\nu \rho  L(x-l)\log (L/K)}{2\pi \sigma }
\int_0^{T} \frac{(T-s)^{1/2}}{s^{3/2}} \exp 
\left (-\frac{c_2(T-s, L/K) + c_2(s, L/e^x)}{2}\right )ds, \nn\\
c_1(t,x)&=& {\log (L^2/e^x K)+\sigma^2 t/2\over \sigma\sqrt{t}}, \ \ 
c_2(t, y) \ = \ \left( \frac{\log y + \sigma ^2t/2}{\sigma \sqrt{t}}\right) ^2. 
\end{eqnarray}
\end{thm}
\begin{proof}
See Appendix \ref{sec_proof_semigroup_th2}.
\end{proof}
\subsection{Numerical Example}
Finally, applying the our approximation formulas in 
Theorem \ref{DOBarrier2} and Theorem \ref{DOBarrier},
we present numerical experiments for European down-and-out barrier call prices.
First, let us denote $u^0=C_\mathrm {Barrier}^{BS}(T,S)$ and $v^0_1=v^0_1(0,\log S)$. 
Then, we see 
\begin{eqnarray*}
C_\mathrm {Barrier}^{SV,\varepsilon}(T,S) \simeq u^0+\varepsilon v^0_1.
\end{eqnarray*} 
In the following we report the results of the numerical experiments, where
the numbers in the parentheses show the error rates (\%)
relative to the benchmark prices of $C_\mathrm {Barrier}^{SV,\varepsilon}(T,S)$;
they are computed by Monte-Carlo simulations with 100,000 time steps and 1,000,000 trials.
We check the accuracy of our approximations
by changing the model parameters. Case 1--6 show the results
for the stochastic volatility model with drifts of the underlying price process or/and
the volatility process (\ref{SVdrift}), 
while Case 7 shows the result for the stochastic volatility model with no drifts  (\ref{no_drift}). 
There, we apply the formula in Theorem \ref{DOBarrier2} to Case 1--6 and 
the formula in Theorem \ref{DOBarrier} to Case 7, respectively.

Apparently, 
\vspace{1mm}
our approximation formula $u^0+\varepsilon v^0_1$ improves the accuracy against $C_\mathrm {Barrier}^{SV,\varepsilon}(T,S)$, 
and it is observed that $\varepsilon v^0_1$ accurately compensates for the difference between 
$C_\mathrm {Barrier}^{SV,\varepsilon}(T,S)$ 
and $C_\mathrm {Barrier}^{BS}(T, S)$, which confirms the validity of our method.\\

\begin{enumerate}
\item
\beas
&&
S=100,\ \sigma=0.15,\ c=0.01,\ q=0.0,\varepsilon \nu=0.2,\ \rho=-0.5, \\
&&
\varepsilon \lambda=0.00,\ \theta=0.00,\
L=95,\ T=0.5,\  K=100,\ 102,\ 105.
\eeas
 \begin{table}[ht]
\begin{center}
\caption{Down-and-Out Barrier Option
}
\label{fig1}
\begin{tabular}{|c|c|c|c| } \hline
{Strike} & \multicolumn{1}{|c|}{Benchmark} & \multicolumn{1}{|c}{Our Approximation ($u^0+\varepsilon v^0_1$)}  
& \multicolumn{1}{|c|}{Barrier Black-Scholes ($u^0$)} 
  \\\hline\hline
100  & 3.468 & 3.466 (-0.05\%) & 3.495 (0.80\%) \\
102  & 2.822 & 2.822 (0.00\%) & 2.866 (1.57\%) \\
105 & 1.986 & 1.986 (0.01\%) & 2.052 (3.36\%) \\
\hline
\end{tabular}
\end{center}
\end{table}
\
\\
\item
\beas
&&
S=100,\ \sigma=0.15,\  c=0.01, \ q=0.0, \varepsilon \nu=0.35,\ \rho=-0.7, \\
&&
\varepsilon \lambda=0.00,\ \theta=0.00,\
L=95,\ T=0.5,\  K=100,\ 102,\ 105.
\eeas
 \begin{table}[ht]
\begin{center}
\caption{Down-and-Out Barrier Option}
\label{fig2}
\begin{tabular}{|c|c|c|c| } \hline
{Strike} & \multicolumn{1}{|c|}{Benchmark} & \multicolumn{1}{|c}{Our Approximation ($u^0+\varepsilon v^0_1$)}  
& \multicolumn{1}{|c|}{Barrier Black-Scholes ($u^0$)} 
  \\\hline\hline
100  & 3.421 & 3.423 (0.07\%) & 3.495 (2.18\%) \\
102  & 2.753 & 2.757 (0.18\%) & 2.866 (4.13\%) \\
105 & 1.885 & 1.890 (0.23\%) & 2.052 (8.88\%) \\
\hline
\end{tabular}
\end{center}
\end{table}
\item
\beas
&&
S=100,\ \sigma=0.15,\  c=0.05, \ q=0.0, \varepsilon \nu=0.35,\ \rho=-0.7,\\
&&
\varepsilon \lambda=0.00,\ \theta=0.00,\
L=95,\ T=0.5,\  K=100,\ 102,\ 105.
\eeas
 \begin{table}[ht]
\begin{center}
\caption{Down-and-Out Barrier Option}
\label{fig3}
\begin{tabular}{|c|c|c|c| } \hline
{Strike} & \multicolumn{1}{|c|}{Benchmark} & \multicolumn{1}{|c}{Our Approximation ($u^0+\varepsilon v^0_1$)}  
& \multicolumn{1}{|c|}{Barrier Black-Scholes ($u^0$)} 
  \\\hline\hline
100  & 4.352 & 4.349 (-0.07\%) & 4.399 (1.06\%) \\
102  & 3.585 & 3.586 (0.02\%) & 3.665 (2.24\%) \\
105 & 2.560 & 2.563 (0.11\%) & 2.696 (5.31\%) \\
\hline
\end{tabular}
\end{center}
\end{table}
\ \\
\item
\beas
&&
S=100,\ \sigma=0.15,\  c=0.05, \ q=0.1,\varepsilon \nu=0.2,\ \rho=-0.5,\\
&&
\varepsilon \lambda=0.00,\ \theta=0.00,\
L=95,\ T=0.5,\  K=100,\ 102,\ 105.
\eeas
\begin{table}[ht]
\begin{center}
\caption{Down-and-Out Barrier Option}
\label{fig4}
\begin{tabular}{|c|c|c|c| } \hline
{Strike} & \multicolumn{1}{|c|}{Benchmark} & \multicolumn{1}{|c}{Our Approximation ($u^0+\varepsilon v^0_1$)}  
& \multicolumn{1}{|c|}{Barrier Black-Scholes ($u^0$)} 
  \\\hline\hline
100  & 2.231 & 2.224 (-0.31\%) & 2.268 (1.64\%) \\
102  & 1.758 & 1.754 (-0.27\%) & 1.812 (3.02\%) \\
105 & 1.172 & 1.168 (-0.31\%) & 1.243 (6.05\%) \\
\hline
\end{tabular}
\end{center}
\end{table}
\ \\
\item
\beas
&&
S=100,\ \sigma=0.15,\  c=0.01, \ q=0.0, \varepsilon \nu=0.2,\ \rho=-0.5,\\
&&
\varepsilon \lambda=0.2,\ \theta=0.25,\
L=95,\ T=0.5,\  K=100,\ 102,\ 105.
\eeas
\begin{table}[ht]
\begin{center}
\caption{Down-and-Out Barrier Option}
\label{fig5}
\begin{tabular}{|c|c|c|c| } \hline
{Strike} & \multicolumn{1}{|c|}{Benchmark} & \multicolumn{1}{|c}{Our Approximation ($u^0+\varepsilon v^0_1$)}  
& \multicolumn{1}{|c|}{Barrier Black-Scholes ($u^0$)} 
  \\\hline\hline
100  & 3.523 & 3.517 (-0.16\%) & 3.495 (-0.77\%) \\
102  & 2.891 & 2.888 (-0.09\%) & 2.866 (-0.85\%) \\
105 & 2.066 & 2.065 (-0.06\%) & 2.052 (-0.64\%) \\
\hline
\end{tabular}
\end{center}
\end{table}

\item
\beas
&&
S=100,\ \sigma=0.15,\  c=0.01, \ q=0.0,\varepsilon \nu=0.2,\ \rho=-0.5,\\
&&
\varepsilon \lambda=0.5,\ \theta=0.25,\
L=95,\ T=0.5,\  K=100,\ 102,\ 105.
\eeas
\begin{table}[ht]
\begin{center}
\caption{Down-and-Out Barrier Option}
\label{fig6}
\begin{tabular}{|c|c|c|c| } \hline
{Strike} & \multicolumn{1}{|c|}{Benchmark} & \multicolumn{1}{|c}{Our Approximation ($u^0+\varepsilon v^0_1$)}  
& \multicolumn{1}{|c|}{Barrier Black-Scholes ($u^0$)} 
  \\\hline\hline
100  & 3.587 & 3.594 (0.20\%) & 3.495 (-2.55\%) \\
102  & 2.976 & 2.987 (0.39\%) & 2.866 (-3.68\%) \\
105 & 2.170 & 2.183 (0.59\%) & 2.052 (-5.41\%) \\
\hline
\end{tabular}
\end{center}
\end{table}
\item
\beas
&&
S=100,\ \sigma=0.15,\  c=0.0, \ q=0.0, \varepsilon \nu=0.2,\ \rho=-0.5,\\
&&
\varepsilon \lambda=0.0,\ \theta=0.0,\
L=95,\ T=0.5,\  K=100,\ 102,\ 105.
\eeas
 \begin{table}[ht]
\begin{center}
\caption{Down-and-Out Barrier Option}
\label{fig7}
\begin{tabular}{|c|c|c|c| } \hline
{Strike} & \multicolumn{1}{|c|}{Benchmark} & \multicolumn{1}{|c}{Our Approximation ($u^0+\varepsilon v^0_1$)}  
& \multicolumn{1}{|c|}{Barrier Black-Scholes ($u^0$)} 
  \\\hline\hline
100  & 3.261 & 3.258 (-0.09\%) & 3.290 (0.90\%) \\
102  & 2.640 & 2.639 (-0.02\%) & 2.686 (1.78\%) \\
105 & 1.841 & 1.841 (0.01\%) & 1.911 (3.77\%) \\
\hline
\end{tabular}
\end{center}
\end{table}
\end{enumerate}

\section{Conclusion}
This paper has proposed an approximation scheme for barrier option prices
by applying a new
semi-group expansion to the Cauchy-Dirichlet problem 
in the second order parabolic partial differential equations (PDEs). 
As an application, we have derived
a semi-group expansion formula 
under a certain type of stochastic volatility model 
and confirmed the validity of our method through numerical examples. 
Developing concrete computational schemes
under various models is our next research topic.

\section{Appendix A: Proof of Theorem 2,3,4,6}\label{sec_appendix}
\subsection{Proof of Theorem \ref{th_classical_vk}}\label{sec_proof_th2}
First, by the definition of $v^0_k$, we easily get 
$v^0_k(T, x) = 0$ for $x\in D$ and 
$v^0_k(t, x) = 0$ for $(t, x)\in [0, T]\times \partial D$. 

Next, fix any $x\in D$. 
By the Markov property, we have 
\begin{eqnarray*}
&&J(t\wedge \tau _D(X^{0, x}))v^0_k\left (t\wedge \tau _D(X^{0, x}), X^{0, x}_{t\wedge \tau _D(X^{0, x})}\right ) \ =\  
J(t)v^0_k(t, X^{0, x}_t)1_{\{\tau _D(X^{0, x})\geq t\}}\\
&=& 
\E \left[ \int ^{T\wedge \tau _D(X^{0, x})}_t
J(r)g^0_k(r, X^{0, x}_r)dr \ \Big| \ \mathcal {F}_t\right] 1_{\{\tau _D(X^{0, x})\geq t\}}\\
&=&
\E \left[ \int ^{T\wedge \tau _D(X^{0, x})}_{0}
J(r)g^0_k(r, X^{0, x}_r)dr \ \Big| \ \mathcal {F}_t\right] - 
\int ^{t\wedge \tau _D(X^{0, x})}_{0}J(r)g^0_k(r, X^{0, x}_r)dr 
\end{eqnarray*}
for each $t\in [0, T]$, where $J(r) = \exp \left ( -\int ^r_0c(X^{0, x}_v, 0)dv\right )$ and 
$(\mathcal {F}_r)_r$ is the Brownian filtration. 
This implies that 
\begin{eqnarray*}
M_t := J(t\wedge \tau _D(X^{0, x}))v^0_k\left (t\wedge \tau _D(X^{0, x}), X^{0, x}_{t\wedge \tau _D(X^{0, x})}\right ) + 
\int ^{t\wedge \tau _D(X^{0, x})}_{0}J(r)g^0_k(r, X^{0, x}_r)dr 
\end{eqnarray*}
is a local martingale. On the other hand, applying Ito's formula, we have that 
\begin{eqnarray*}
M_t &=& M_0 + \int ^t_0\left\{ \left( \frac{\partial }{\partial t} + \mathscr {L}^0\right )
v^0_k(r, X^{0, x}_r) + 
g^0_k(r, X^{0, x}_r)\right\}1_{\{\tau _D(X^{0, x})\geq r\}}dr\\
&& + 
\sum ^d_{i, j = 1}\int ^t_0J(r)
\sigma ^{ij}(X^{0, x}_r, 0)
\frac{\partial }{\partial x^i}v^0_k(r, X^{0, x}_r)1_{\{\tau _D(X^{0, x})\geq r\}}dB^j_r
\end{eqnarray*}
for each $t\in [0, T]$. Thus, the uniqueness of decompositions of semimartingales gives us 
\begin{eqnarray*}
\int ^t_0\left\{ \left( \frac{\partial }{\partial t} + \mathscr {L}^0\right )
v^0_k(r, X^{0, x}_r) + 
g^0_k(r, X^{0, x}_r)\right\}1_{\{\tau _D(X^{0, x})\geq r\}}dr = 0, \ \ t\in [0, T]. 
\end{eqnarray*}
Therefore, for each fixed $t\in (0, T)$, 
\begin{eqnarray*}
\frac{1}{h}\int ^{t+h}_t\left\{ \left( \frac{\partial }{\partial t} + \mathscr {L}^0\right )
v^0_k(r, X^{0, x}_r) + 
g^0_k(r, X^{0, x}_r)\right\}1_{\{\tau _D(X^{0, x})\geq r\}}dr = 0 
\end{eqnarray*}
holds for any small enough $h > 0$. 
Since $x\in D$, by letting $h\rightarrow 0$, we obtain 
\begin{eqnarray*}
\left( \frac{\partial }{\partial t} + \mathscr {L}^0\right )
v^0_k(t, x) + 
g^0_k(t, x) = 0. 
\end{eqnarray*}
Finally we prove (\ref {est_v0k}) by mathematical induction. 
When $k = 0$, the assertion is easily obtained by (\ref {SDE_poly_growth}), (\ref {g_expansion}), [F] and [G]. 
Now we assume that (\ref {est_v0k}) holds for $1, 2, \ldots , k - 1$. 
Then, by (\ref {g_expansion}), 
(\ref {def_g}) and [F], we have 
\begin{eqnarray*}
|g^0_k(t, x)| \leq C(1 + |x|^{2m})\sum _{|\alpha |\leq 2}\left( |D_\alpha u^0(t, x)| + 
\sum ^{k-1}_{l = 1}|D_\alpha v^0_l(t, x)|\right) 
\end{eqnarray*}
for some $C, m > 0$, where $\alpha = (i_1, \ldots , i_d)\in \{0, 1, 2, \ldots \}^d$ is a multi-index, 
$|\alpha | = i_1 + \cdots + i_d$ and 
$D_\alpha = \partial ^{|\alpha |}/{(\partial x^1)^{i_1}\cdots (\partial x^d)^{i_d}}$. 
By the induction hypothesis and [G]--[H], we see that 
\begin{eqnarray*}
\sum _{|\alpha |\leq 2}\left( |D_\alpha u^0(t, x)| + 
\sum ^{k-1}_{l = 1}|D_\alpha v^0_l(t, x)|\right) \leq C'M(t)(1 + |x|^{2m'})
\end{eqnarray*}
for some $C', m' > 0$ and $M\in L^1([0, T). dt)$. 
Therefore, we get 
\begin{eqnarray*}
|g^0_k(t, x)| \leq C''M(t)(1 + |x|^{2m''})
\end{eqnarray*}
for some $C'', m'' > 0$. 
Then we obtain 
\begin{eqnarray*}
|v^0_k(t, x)| 
&\leq & 
C''\E \left [\int ^{(T - t)}_0M(t + r)(1 + |X^{0, x}_r|^{2m''})dr\right ]\\
&\leq & 
C''\left( 1 + \E [\sup _{0\leq r\leq T}|X^{0, x}_r|^{2m''}]\right) \int ^T_0M(r)dr \\
&\leq & 
C''C_{m''}T^{m''-1}\left( \int ^T_0M(r)dr\right) (1 + |x|^{2m''}) 
\end{eqnarray*}
by virtue of (\ref {SDE_poly_growth}). 
Thus (\ref {est_v0k}) also holds for $k$. 
Now we complete the proof of Theorem \ref {th_classical_vk}.

\subsection{Proof of Theorem \ref {th_main}}\label{sec_proof}
First, we generalize the definitions of $\tilde{\mathscr {L}}^0_k$, $g^0_k$ and $v^0_k$. 
For $k,n \geq 1$, We define 
\begin{eqnarray*}
\tilde{\mathscr {L}}^\varepsilon _k &=& 
\frac{1}{(k-1)!}\Bigg\{ 
\frac{1}{2}\sum ^d_{i, j = 1}\int ^1_0(1-r)^{k-1}\frac{\partial ^ka^{ij}}{\partial \varepsilon ^k}(x, r\varepsilon )dr\frac{\partial ^2}{\partial x^i\partial x^j}\\&&\hspace{18mm} + 
\sum ^d_{i = 1}\int ^1_0(1-r)^{k-1}\frac{\partial ^kb^i}{\partial \varepsilon ^k}(x, r\varepsilon )dr\frac{\partial }{\partial x^i} - 
\int ^1_0(1-r)^{k-1}\frac{\partial ^kc}{\partial \varepsilon ^k}(x, r\varepsilon )dr
\Bigg\} , \\
g^\varepsilon _n(t, x) &=& \tilde{\mathscr {L}}^\varepsilon _nu^0(t, x) + 
\sum ^{n-1}_{k = 1}\tilde{\mathscr {L}}^{0}_{n-k}v^{0}_k(t, x) + 
\sum ^{n-2}_{k = 1}\varepsilon ^k
\left\{
\tilde{\mathscr {L}}^\varepsilon _nv^0_k(t, x)
+ \sum ^{n-1}_{l = k + 1}\tilde{\mathscr {L}}^0_{n + k - l}v^0_k(t, x)\right\}\\
&&
+\varepsilon ^{n-1}\tilde{\mathscr {L}}^\varepsilon _nv^0_{n-1}(t, x),
\end{eqnarray*}
where $g_1^\varepsilon(t,x)$ and $g_2^\varepsilon(t,x)$ 
are understood as $g_1^\varepsilon(t,x)=\tilde{\mathscr {L}}^\varepsilon _1 u^0(t, x)$
and
$g_2^\varepsilon(t,x)=\tilde{\mathscr {L}}^\varepsilon _2 u^0(t, x) + \tilde{\mathscr {L}}^0 _1 v_1^0(t, x)
+ \varepsilon\tilde{\mathscr {L}}^\varepsilon _2 v^0_1(t, x)$, respectively.

We consider the following Cauchy-Dirichlet problem: 
\begin{eqnarray}\label{eq_PDE2_eps}
\left\{
\begin{array}{ll}
 	-\frac{\partial }{\partial t}v(t, x) - \mathscr {L}^\varepsilon v(t, x) - g^\varepsilon _n(t, x) = 0,& (t, x)\in [0, T)\times D,	\\
 	v(T, x) = 0,& x\in D,	\\
 	v(t, x) = 0,& (t, x)\in [0, T]\times \partial D.
\end{array}
\right.
\end{eqnarray}

For $\varepsilon \neq 0$, 
we define $v^\varepsilon _n = [u^\varepsilon - \{ u^0 + \sum ^{n-1}_{k = 1}\varepsilon ^kv^0_k(t, x)\} ] / \varepsilon ^n$. 
Obviously, we see 
\begin{eqnarray}\label{temp_expansion2}
u^\varepsilon (t, x) = u^0(t, x) + \sum ^{n-1}_{k = 1}\varepsilon ^kv^0_k(t, x) + \varepsilon ^nv^\varepsilon _n(t, x). 
\end{eqnarray}

\begin{prop} \ \label{prop_v}
The function $v^\varepsilon _n$ is a solution of $(\ref {eq_PDE2_eps})$. 
\end{prop}

\begin{proof} 
It is obvious that $v^\varepsilon _n(T, x) = 0$ for $x\in D$ and 
$v^\varepsilon _n(t, x) = 0$ for $(t, x)\in [0, T]\times \partial D$. 
Apply Taylor's theorem 
to $\mathscr {L}^\varepsilon$ in (\ref {eq_PDE}) to observe 
\begin{eqnarray}\label{temp_expansion1}
\mathscr {L}^\varepsilon u^\varepsilon (t, x) = 
\left\{ \mathscr {L}^0 + \sum ^{n-1}_{k = 1}\varepsilon ^k\tilde{\mathscr {L}}^0_k + \varepsilon ^n\tilde{\mathscr {L}}^\varepsilon _n\right\} 
u^\varepsilon (t, x). 
\end{eqnarray}
Since $u^0$ is the solution of (\ref {eq_PDE}) with $\varepsilon = 0$, we get 
\begin{eqnarray}\label{temp_expansion3}
\frac{\partial }{\partial t}u^0(t, x) + 
\mathscr {L}^0u^0(t, x) = 0. 
\end{eqnarray}
Similarly, since $v^0_k$ is a solution of (\ref {eq_PDE2}), we have 
\begin{eqnarray}\label{temp_expansion4}
\frac{\partial }{\partial t}v^0_k(t, x) + 
\mathscr {L}^0v^0_k(t, x) + 
\tilde{\mathscr {L}}^0_ku^0(t, x) + 
\sum ^{k-1}_{l = 1}\tilde{\mathscr {L}}^0_{k-l}v^0_l(t, x) = 0. 
\end{eqnarray}
Combining (\ref {temp_expansion2})--(\ref {temp_expansion4}) and Theorem \ref {th_viscosity}, we obtain 
\begin{eqnarray*}
&&\varepsilon ^n\left\{ \frac{\partial }{\partial t}v^\varepsilon _n(t, x) + 
\mathscr {L}^0 v^\varepsilon _n(t, x)
+ 
\tilde{\mathscr {L}}^\varepsilon _nu^0(t, x) + 
\sum ^{n-1}_{l = 1}\tilde{\mathscr {L}}^0_{n-l}v^0_l(t, x) \right\} \\
&& + 
\sum ^{2n-2}_{k = n + 1}\varepsilon ^k\left\{ 
\tilde{\mathscr {L}}^0_{k-n}v^\varepsilon _n(t, x) + 
\tilde{\mathscr {L}}^\varepsilon _nv^0_{k - n}(t, x) + 
\sum ^{n-1}_{l = k - n + 1}\tilde{\mathscr {L}}^0_{k-l}v^0_l(t, x)\right\}
\\&&
+ \varepsilon ^{2n-1}\left\{\tilde{\mathscr {L}}^0_{n-1}v^\varepsilon _n(t, x) 
+ 
\tilde{\mathscr {L}}^\varepsilon _nv^0_{n-1}(t, x)
\right\}
+ 
\varepsilon ^{2n}\tilde{\mathscr {L}}^\varepsilon _nv^\varepsilon _n(t, x) = 0, 
\end{eqnarray*}
and thus, 
\begin{eqnarray*}
\frac{\partial }{\partial t}v^\varepsilon _n(t, x) + 
\mathscr {L}^\varepsilon v^\varepsilon _n(t, x) + g^\varepsilon _n(t, x) = 0. 
\end{eqnarray*}
This implies the assertion. 
\end{proof}

Set 
\begin{eqnarray*}\label{temp_bdd1}
\tilde{v}^\varepsilon _n(t, x) = \E \left [\int ^{\tau _D(X^{\varepsilon , x})\wedge (T-t)}_0
\exp \left( -\int ^r_0c(X^{\varepsilon , x}_v, \varepsilon )dv\right) g^\varepsilon _n(t + r, X^{\varepsilon , x}_r)dr\right ] . 
\end{eqnarray*}
By [G]--[H], we find that there are $C_n > 0$, $\tilde{m}_n\in \Bbb {N}$ which are independent of $\varepsilon $ and 
the function $M_n\in C([0, T))\cap L^1([0, T), dt)$ determined by $u^0, v^0_1, \ldots , v^0_{n - 1}$ such that 
\begin{eqnarray}\label{poly_g_eps}
|g^\varepsilon _n(t, x)| \leq C_nM_n(t)(1 + |x|^{2\tilde{m}_n}). 
\end{eqnarray}
The inequalities (\ref {SDE_poly_growth}) and (\ref {poly_g_eps}) imply 
\begin{eqnarray}\label{poly_v_eps}
|\tilde{v}^\varepsilon _n(t, x)| \leq C'_n\int ^T_tM_n(r)dr(1 + |x|^{2\tilde{m}_n})
\end{eqnarray}
for some $C'_n > 0$ which is also independent of $\varepsilon $.

\begin{prop} \ \label{prop_uniqueness_v_eps}
$v^\varepsilon _n = \tilde{v}^\varepsilon _n$. 
\end{prop}

\begin{proof} 
The assertion is easily obtained by the similar argument to the one in Theorem 5.1.9 in \cite {Lamberton-Lapeyre}. 
\end{proof}

\begin{proof}[Proof of Theorem \ref {th_main}]
By (\ref {temp_expansion2}) and Proposition \ref {prop_uniqueness_v_eps}, we have 
$u^\varepsilon (t, x) - 
(u^0(t, x) + \sum ^{n-1}_{k = 1}\varepsilon ^kv^0_k(t, x)) = 
\varepsilon ^n\tilde{v}^\varepsilon _n(t, x)$. 
Our assertion is now immediately obtained by the inequality (\ref {poly_v_eps}). 
\end{proof}

\subsection{Proof of Theorem \ref{th_form_semigroup}}\label{proof_form_semigroup}
\begin{enumerate}
\item
Firstly, let us consider the case for $k = 1$.
Let $g\in \mathcal {H}^{m, 1}$. 
Observe that
\begin{eqnarray*}
&&\int ^{(T - t)\wedge \tau _D(X^{0, x})}_0
\exp \left( -\int ^r_0c(X^{0, x}_v,0)dv\right)g(t + r, X^{0, x}_r)dr\\
&=& 
\int ^{T - t}_0
\exp \left( -\int ^r_0c(X^{0, x}_v,0)dv\right)g(t + r, X^{0, x}_r)1_{\{ \tau _D(X^{0, x})\geq r\} }dr, 
\end{eqnarray*}
and we obtain
\begin{eqnarray*}
&&\E \left [\int ^{(T - t)\wedge \tau _D(X^{0, x})}_0
\exp \left( -\int ^r_0c(X^{0, x}_v,0)dv\right)g(t + r, X^{0, x}_r)dr\right ]\\
&=& 
\int ^{T - t}_0\E \left [
\exp \left( -\int ^r_0c(X^{0, x}_v,0)dv\right)g(t + r, X^{0, x}_r)1_{\{ \tau _D(X^{0, x})\geq r\} }\right ]dr\\
&=& 
\int ^{T - t}_0P^D_rg(t + r, \cdot )(x)dr. 
\end{eqnarray*}
Thus, under the assumption [H], we see 
\begin{eqnarray}
v^0_1(T - t, x) &=& 
\E \left[ \int ^{t}_0
\exp \left( -\int ^r_0 c(X^{0, x}_v,0) dv
\right) g^0_1(T-t + r, X^{0, x}_r) 1_{\{\tau _D(X^{0, x})\geq r\}} dr\right ]\nn\\
&=&
\int ^t_0P^D_r\tilde{\mathscr {L}}^0_1u^0(T - t + r, \cdot )(x) 
dr \nn 
\label{form_semi_0}
\\ 
&=& 
\int ^t_0P^D_r\tilde{\mathscr {L}}^0_1P^D_{t - r}f(x)dr = 
\int ^t_0P^D_{t - r}\tilde{\mathscr {L}}^0_1P^D_rf(x)dr. \label{form_semi_1}
\end{eqnarray}
Thus, we have the assertion for $k = 1$. 
\item
If the assertion holds for $1, \ldots , k - 1$, then 
\begin{eqnarray*}
&&v^0_k(T - t, x) = \int ^t_0P^D_{t_0}\{ \tilde{\mathscr {L}}^0_ku^0 + 
\sum ^{k - 1}_{l = 1}\tilde{\mathscr {L}}^0_{k - l}v^0_l\} (T - t + t_0, \cdot )(x)dt_0 \\
&=& 
\int ^t_0P^D_{t - t_0}\tilde{\mathscr {L}}^0_kP^D_{t_0}f(x)dt_0\nonumber\\&& + 
\sum ^{k - 1}_{l = 1}
\sum ^l_{m = 1}\sum _{(\beta ^i)^m_{i=1}\subset \Bbb {N}^m, \sum _i\beta ^i = l}
\int ^t_0\int ^{t_0}_0\int ^{t_1}_0\cdots \int ^{t_{l - 1}}_0\nonumber\\&&\hspace{35mm}
P^D_{t - t_0}\tilde{\mathscr {L}}^0_{k - l}
P^D_{t_0 - t_1}\tilde{\mathscr {L}}^0_{\beta ^1}P^D_{t_1 - t_2}\tilde{\mathscr {L}}^0_{\beta ^2}
\cdots P^D_{t_{l - 1} - t_l}\tilde{\mathscr {L}}^0_{\beta ^l}P^D_{t_l}f(x)
dt_l\cdots dt_1dt_0\nonumber\\
&=& 
\int ^t_0P^D_{t - t_0}\tilde{\mathscr {L}}^0_kP^D_{t_0}f(x)dt_0\\&& + 
\sum ^{k}_{l = 2}
\sum ^l_{m = 1}\sum _{(\beta ^i)^m_{i=1}\subset \Bbb {N}^m, \sum _i\beta ^i = k}
\int ^t_0\int ^{t_1}_0\int ^{t_2}_0\cdots \int ^{t_{l - 1}}_0\\&&\hspace{35mm}
P^D_{t - t_1}\tilde{\mathscr {L}}^0_{\beta ^1}
P^D_{t_1 - t_2}\tilde{\mathscr {L}}^0_{\beta ^2}P^D_{t_2 - t_3}\tilde{\mathscr {L}}^0_{\beta ^3}
\cdots P^D_{t_{l - 1} - t_l}\tilde{\mathscr {L}}^0_{\beta ^l}P^D_{t_l}f(x)
dt_l\cdots dt_1\\
&=& 
\sum ^k_{l = 1}\sum _{(\beta ^i)^l_{i=1}\subset \Bbb {N}^l, \sum _i\beta ^i = k}
\int ^t_0\int ^{t_1}_0\cdots \int ^{t_{l - 1}}_0
P^D_{t - t_1}\tilde{\mathscr {L}}^0_{\beta ^1}P^D_{t_1 - t_2}\tilde{\mathscr {L}}^0_{\beta ^2}
\cdots P^D_{t_{l - 1} - t_l}\tilde{\mathscr {L}}^0_{\beta ^l}P^D_{t_l}f(x)dt_l\cdots dt_1. 
\end{eqnarray*}
Thus, our assertion is also true for $k$. Then we complete the proof of 
Proposition \ref {th_form_semigroup} 
by mathematical induction. 
\end{enumerate}

\subsection{Proof of Theorem \ref {DOBarrier}}\label{sec_proof_semigroup_th2}
By the asymptotic expansion in Section 3
and Theorem \ref{th_form_semigroup} with $k=1$,
we see that the expansion 
\begin{eqnarray*}
C_\mathrm {Barrier}^{SV,\varepsilon}(T,e^x) = 
C_\mathrm {Barrier}^{BS,\varepsilon}(T,e^x) +\varepsilon v^0_1(0, x) + O(\varepsilon ^2)
\end{eqnarray*}
holds with 
\begin{eqnarray}\label{eg_v1}
v^0_1(t, x) = \int ^{T - t}_0P^{D}_{T - t - r}\tilde{\mathscr {L}}_1^0 P^{D}_r\bar{f} (x)dr. 
\end{eqnarray}
Then, we have the following proposition for an expression of $v^0_1(0, x)$.
The proof 
is given in Section \ref {sec_proof_eg_prop1}. 
\begin{prop} \ \label{Barrier_Semi-group}
\begin{eqnarray*}
v^0_1(0, x) = \frac{T}{2}\nu \rho \sigma^2\frac{\partial ^2}{\partial x \partial \sigma }P^{D}_{T}\bar{f}(x) - 
\frac{1}{2}\E [(T - \tau _D(X^{0, x}))\nu \rho \sigma^2 \frac{\partial ^2}{ \partial x \partial \sigma}P^{D}_{T - \tau _D(X^{0, x})}\bar{f}(l)1_{\{\tau _D(X^{0, x}) < T\}}]. 
\end{eqnarray*}
\end{prop}

We remark that the expectation in the above equality can be represented as
\begin{eqnarray}
&&\frac{1}{2}\E [(T - \tau _D(X^{0, x}))\nu \rho \sigma^2 \frac{\partial ^2}{ \partial x \partial \sigma}P^{D}_{T - \tau _D(X^{0, x})}\bar{f}(l)1_{\{\tau _D(X^{0, x}) < T\}}]\nn\\
&=&
\int_0^{T} {(T-s)\over 2}\nu \rho \sigma^2 \frac{\partial ^2}{ \partial x \partial \sigma}P^{D}_{T -s}\bar{f}(l) h(s,x-l) ds,\label{f_2}
\end{eqnarray} 
where 
$h(s,x-l)$ is the density function of the first hitting time to $l$ defined by 
\be
h(s,x-l) = {-(l-x)\over \sqrt{2\pi \sigma^2 s^3}}\exp\l(-{\l\{l-x+\sigma^2s/2 \r\}^2\over 2\sigma^2 s}\r).
\ee
Now we evaluate 
\begin{eqnarray}\nonumber 
\nu \rho \sigma^2 \frac{\partial ^2}{ \partial x \partial \sigma}P^{D}_{t}\bar{f}(x)&=&\nu \rho \sigma^2 \frac{\partial ^2}{ \partial x \partial \sigma} C^{BS}(t,e^x)
-\nu \rho \sigma^2 \frac{\partial ^2}{ \partial x \partial \sigma} 
\left\{ \l(\frac{e^x}{L} \r)C^{BS}\l(t,\frac{L^2}{e^x} \r) \right\} .
\label{f_20000}
\end{eqnarray}
Note that
\begin{eqnarray}
\frac{\partial}{\partial \sigma}C^{BS}\l(t,{e^x} \r)= e^x n(d_1(t,x))\sqrt{t},
\end{eqnarray}
and
\begin{eqnarray}
\frac{\partial}{\partial \sigma}\left\{ 
\l(\frac{e^x}{L} \r)C^{BS}\l(t,\frac{L^2}{e^x} \r) \right\} = L n(c_1(t,x))\sqrt{t}.
\end{eqnarray}
Then we have
\begin{eqnarray}
\nu \rho\sigma^2 \frac{\partial^2}{\partial x  \partial \sigma}C^{BS}\l(t,{e^x} \r)
&=&\nu \rho\sigma^2 e^x n(d_1(t,x))\sqrt{t}\l\{ 1-\frac{d_1(t,x)}{\sigma \sqrt{t}}\r\}\nn\\
&=& -\nu \rho \sigma e^x n(d_1(t,x))d_2(t,x)\label{f_30000}
\end{eqnarray}
and
\begin{eqnarray}
\nu \rho\sigma^2 \frac{\partial^2}{\partial x  \partial \sigma}\left\{ 
\l(\frac{e^x}{L} \r)C^{BS}\l(t,\frac{L^2}{e^x} \r)\right\} 
&=&\nu \rho \sigma Ln(c_1(t,x))c_1(t,x). \label{f_40000}
\end{eqnarray}
Combining (\ref {f_20000}), (\ref {f_30000}) and (\ref {f_40000}), we get 
\begin{eqnarray}
\nu \rho \sigma^2 \frac{\partial ^2}{ \partial x \partial \sigma}P^{D}_{t}\bar{f}(x)
&=&\nu \rho  \sigma \l\{ e^x n(d_1(t,x))(-d_2(t,x))-Ln(c_1(t,x))c_1(t,x)\r\}. \label{f_10000}
\end{eqnarray}
Substituting (\ref {f_10000}) into (\ref {f_2}), we have
\begin{eqnarray*}
\nu \rho \sigma^2 \frac{\partial ^2}{ \partial x \partial \sigma}P^{D}_{t}\bar{f}(l)
&=&\nu \rho  \sigma L n(d_1(t,l))(-d_2(t,l))-\rho  \sigma Ln(c_1(t,l))c_1(t,l)\nn\\
&=&\nu \rho  \sigma L n(d_1(t,l))(-(d_1(t,l)+d_2(t,l)))\nn\\
&=&\nu \rho  \sigma \frac{1}{\sqrt{2\pi}} \exp \l(-\frac{\{l-\log K+\frac{1}{2}\sigma^2 {t}\}^2}{2\sigma^2 {t}} \r)
\l(\frac{-2(l-\log K)}{\sigma \sqrt{t}}\r).
\end{eqnarray*}
Thus we obtain
\begin{eqnarray}
&&-\frac{1}{2}\E [(T - \tau _D(X^{0, x}))\nu \rho \sigma^2 \frac{\partial ^2}{ \partial x \partial \sigma}P^{D}_{T - \tau _D(X^{0, x})}\bar{f}(l)1_{\{\tau _D(X^{0, x}) < T\}}]\nn\\
&=&-\int_0^{T} {(T-s)\over 2}\nu \rho  \sigma L \frac{1}{\sqrt{2\pi}} e^{-\frac{\l\{l-\log K+\frac{1}{2}\sigma^2 {(T-s)}\r\}^2}{2\sigma^2 {(T-s)}} }
\l(\frac{-2(l-\log K)}{\sigma \sqrt{T-s}}\r)\nn\\
&&
\times{-(l-x)\over \sqrt{2\pi \sigma^2s^3}}e^{-{\l\{(l-x)+(\sigma^2/2) s \r\}^2\over 2\sigma^2 s}}ds\nn\\
&=&  \frac{\nu \rho L(x-l)\log(L/K)}{{ 2\pi }\sigma }
\int_0^{T} \frac{(T-s)^{1/2}}{s^{3/2}} 
\exp \left (-\frac{c_2(T-s, L/K) + c_2(s, L/e^x)}{2}\right )
ds.\nn \\
\label{f_3}
\end{eqnarray}
By Proposition \ref {Barrier_Semi-group}, (\ref{f_2}), (\ref {f_10000}) and (\ref{f_3}), we reach the assertion.

\subsubsection{Proof of Proposition \ref {Barrier_Semi-group}}\label{sec_proof_eg_prop1}
First, we notice the following relation: 
\begin{eqnarray}
\tilde{\mathscr {L}}_1^0 P^{D}_t\bar{f} (x)
=\nu \rho \sigma^3 t \l(\frac{\partial^3}{\partial x^3}-\frac{\partial^2}{\partial x^2} \r)P^{D}_{t}\bar{f}(x).
\end{eqnarray} 
Then, using the relations $\mathscr {L}^0 \tilde{\mathscr {L}}_1^0P^D_t\bar{f}(x)=\tilde{\mathscr {L}}_1^0 \mathscr {L}^0P^D_t\bar{f}(x)$ 
and 
\begin{eqnarray*}
\l(\frac{\partial}{\partial t}+ \mathscr {L}^0 \r) P^{D}_{T-t}\bar{f} (x)=0,
\end{eqnarray*}
we get 
\begin{eqnarray} 
\l(\frac{\partial}{\partial t}+ \mathscr {L}^0 \r)\frac{T-t}{2}  \tilde{\mathscr {L}}_1^0 P^{D}_{T-t}\bar{f} (x)
&=&-\tilde{\mathscr {L}}_1^0 P^{D}_{T-t}\bar{f}(x).
\end{eqnarray} 
Also, we have 
\begin{eqnarray}
\l(\frac{\partial}{\partial t}+ \mathscr {L}^0 \r)
\int ^{T-t}_0 P^{D}_{T -t- r}\left(\nu \rho \sigma^2 \frac{\partial ^2}{\partial x \partial \sigma }P^{D}_r\bar{f}\right) (x)dr
 = -\tilde{\mathscr {L}}_1^0 P^{D}_{T-t}\bar{f}(x), \ \ x\in (l, \infty ).
\end{eqnarray}
Therefore, the function 
\begin{eqnarray}\label{def_eta}
\eta(t,x)= \int ^{T-t}_0 P^{D}_{T-t - r}\left(\nu \rho \sigma^2 \frac{\partial ^2}{\partial x \partial \sigma }P^{D}_r\bar{f}\right) (x)dr 
- \frac{T-t}{2}  \tilde{\mathscr {L}}_1^0 P^{D}_{T-t}\bar{f} (x)
\end{eqnarray}
satisfies the following PDE 
\begin{eqnarray*} 
\left\{
\begin{array}{ll}
 	\l(\frac{\partial}{\partial t}+ \mathscr {L}^0 \r)\eta(t,x)=0,& (t, x)\in [0, T)\times (l, \infty ),	\\
 	\eta(T,x) = 0, & x\in [l, \infty ), \\
 	\eta(t,l) = -\frac{T-t}{2}  \tilde{\mathscr {L}}_1^0 P^{D}_{T-t}\bar{f} (l),& 	t\in [0, T). 
\end{array}
\right.
\end{eqnarray*} 
Then Theorem 6.5.2 in \cite{Friedman2} implies 
\begin{eqnarray}\label{rep_eta}
 \eta(0,x)=- \frac{1}{2}\E [(T - \tau _D(X^{0, x}))\nu \rho \sigma^2 \frac{\partial ^2}{ \partial x \partial \sigma}P^{D}_{T - \tau _D(X^{0, x})}\bar{f}(l)1_{\{\tau _D(X^{0, x}) < T\}}].  
\end{eqnarray} 
By (\ref {def_eta}) and (\ref {rep_eta}), we get the assertion. \qed 
\noindent\\

\section{Appendix B: Generalization 
}\label{sec_general}
This section generalizes the results given in Section 2 and Section 3
to treat more general cases covered in the main text.

Particularly, let $d'\in \{1, \ldots  , d\}$, and
we regard $X^{\varepsilon , x, i}_t$ as logarithm of the underlying asset prices 
for $i \leq d'$, 
and as parameter processes (e.g. a stochastic volatility and a stochastic interest rate) for $i > d'$. 
Also, we assume $I \subset [0, \infty )$ in this section for a technical reason introduced later. 

Let $(\Omega , \mathcal {F}, (\mathcal {F}_t)_t, P)$ be a filtered space 
equipped with a standard Brownian motion $(B_t)_t$. 
Set 
\begin{eqnarray}
\hat{b}^i(y, \varepsilon ) &=& 
\left\{
\begin{array}{ll}
 	y^i\left\{ b^i(\pi (y), \varepsilon ) + \frac{1}{2}\sum ^d_{j = 1}(\sigma ^{ij}(\pi (y), \varepsilon ))^2\right\}, & i \leq d',	\\
 	b^i(\pi (y), \varepsilon ), & i > d', 
\end{array}
\right.\\
\hat{\sigma }^{ij}(y, \varepsilon ) &=& 
\left\{
\begin{array}{ll}
 	y^i\sigma ^{ij}(\pi (y), \varepsilon ), & i \leq d',	\\
 	\sigma ^{ij}(\pi (y), \varepsilon ), & i > d', 
\end{array}
\right.
\end{eqnarray}
where $\pi (y) = (\log y^1, \ldots , \log y^{d'}, y^{d' + 1}, \ldots , y^d) \in \Bbb {R}^d$.

Next, we introduce new assumptions for the generalization: $[A']$--$[C']$ and $[G']$--$[H']$ below with $[D], [F]$ 
in Section 2 and Section 3, respectively
are necessary for the 
generalization (Theorem \ref{th_general} below)
of the asymptotic expansion. 

\begin{description}
 \item[\mbox{[A']}] \ 
For each $\varepsilon \in I$ it holds that $\sigma ^{ij}(\cdot , \varepsilon ), b^i(\cdot , \varepsilon )\in \mathcal {L}$,
and that 
$\hat{\sigma }^{ij}(\cdot , \varepsilon ), \hat{b}^i(\cdot , \varepsilon )$ 
and $c(\pi (\cdot ), \varepsilon )$ are also in $\mathcal {L}$.
Here, $\mathcal {L}$ is defined in the assumption $[A]$ of Section 2,
that is the set of
locally Lipschitz continuous functions defined on $\Bbb {R}^d$.

Moreover, there exists a solution $(X^{\varepsilon , x}_t)_t$ of SDE (\ref {eq_SDE_epsilon}) and 
for any $m > 0$ there are $m', C > 0$ such that 
\begin{eqnarray}\label{ineq_exp}
\sup _{0\leq r\leq t}\E [|Y^{\varepsilon , y}_r|^{2m}] \leq Ct^{m - 1}(1 + |y|^{2m'}) , \ \ 
(t, y)\in [0, T]\times [0, \infty )^{d'}\times \Bbb {R}^{d - d'}, \ \varepsilon \in  I , 
\end{eqnarray}
where 
\bea
Y^{\varepsilon , y}_t &=& \iota (X^{\varepsilon , \pi (y)}_t),\\
\iota (x) &=& (e^{x^1}, \ldots , e^{x^{d'}}, x^{d' + 1}, \ldots , x^d)\in \Bbb {R}^d. \nn
\eea
\end{description}
\begin{rem}
Note that Ito's formula implies that $(Y^{\varepsilon , y}_t)_t$ is a solution of 
\begin{eqnarray*}
\left\{
\begin{array}{l}
 	dY^{\varepsilon , y}_t = \hat{b}(Y^{\varepsilon , y}_t, \varepsilon )dt + \hat{\sigma }(Y^{\varepsilon , y}_t, \varepsilon )dB_t,\\
 	Y^{\varepsilon , y}_0 = y. 
\end{array}
\right.
\end{eqnarray*}
\end{rem}

\begin{description}
\item[\mbox{[B']}] \ The function $f(x)$ is represented by the continuous function 
$\hat{f} : \Bbb {R}^d\longrightarrow \Bbb {R}$ as $f(x) = \hat{f}(\iota (x))$. 
There exists $C_{\hat{f}} > 0$ such that $|\hat{f}(y)|^2 \leq C_{\hat{f}}(1 + |y|^{2m})$, $y\in \Bbb {R}^d$. 
Moreover, $f(x) = 0$ on $\Bbb {R}^d\setminus D$. 
 \item[\mbox{[C']}] \ 
\ In addition to the condition [C] ($c(x, \varepsilon )\geq 0$ and $c(\cdot , \epsilon )\in \mathcal {L}$), 
there is a constant $A^\varepsilon _2 > 0$ such that 
\begin{eqnarray*}
|\hat{\sigma }^{ij}(y, \varepsilon )|^2 + 
|\hat{b}^{i}(y, \varepsilon )|^2 
&\leq & 
A^\varepsilon _2(1 + |y|^2), \ \ i, j = 1, \ldots , d, \\
c(x, \varepsilon )^2 
&\leq & 
A^\varepsilon _2(1 + |\iota (x)|^{2m}). 
\end{eqnarray*}
\end{description}

\begin{rem}
We remark that 
$[C']$ implies 
\begin{eqnarray}\label{gen_linear_growth}
|\sigma ^{ij}(x, \varepsilon )|^2 + 
|b^{i}(x, \varepsilon )|^2 
\leq 
A^\varepsilon _3(1 + |y|^2), \ \ i, j = 1, \ldots , d
\end{eqnarray}
for some $A^\varepsilon _3 > 0$. 
\end{rem}

We note that Theorem 3.1 in \cite{Rubio} no longer works for the PDE (\ref{eq_PDE}): 
\begin{eqnarray*}
\left\{
\begin{array}{ll}
 	\frac{\partial }{\partial t}u^\varepsilon (t, x) + \mathscr {L}^\varepsilon u^\varepsilon (t, x) = 0,& (t, x)\in [0, T)\times D,	\vspace{1mm}\\
	u^\varepsilon (T, x) = f(x),& x\in D,	\\
	u^\varepsilon (t, x) = 0,& (t, x)\in [0, T]\times \partial D 
\end{array}
\right.
\end{eqnarray*}
under $[A']$--$[B']$. 
Now, we focus on the generator 
\begin{eqnarray*}
\hat{\mathscr {L}}^\varepsilon &=& 
\frac{1}{2}\sum ^d_{i, j = 1}\hat{a}^{ij}(y, \varepsilon )\frac{\partial ^2}{\partial y^i\partial y^j} + 
\sum ^d_{i = 1}\hat{b}^i(y, \varepsilon )\frac{\partial }{\partial y^i} - c(\pi (y), \varepsilon ), 
\end{eqnarray*}
rather than $\mathscr {L}^\varepsilon $, 
where $\hat{a}^{ij} = \sum ^d_{k = 1}\hat{\sigma }^{ik}\hat{\sigma }^{jk}$. 
Moreover, define 
\begin{eqnarray*}
\hat{D} = \{ y\in \Bbb {R}^d\ ; \ y^i > 0, \ i = 1, \ldots , d'\ \mbox{and} \ \pi (y) \in D\} 
\end{eqnarray*}
and 
$\hat{u}^\varepsilon (t, y) = u^\varepsilon (t, \pi (y))$ $((t, y)\in [0, T]\times \hat{D})$, 
$0$ $((t, y)\in [0, T]\times \partial \hat{D})$. 
Then the function $\hat{u}^\varepsilon $ is expected to be the solution of 
\begin{eqnarray}\label{eq_PDE_S}
\left\{
\begin{array}{ll}
 	-\frac{\partial }{\partial t}\hat{u}^\varepsilon (t, y) - \hat{\mathscr {L}}^\varepsilon \hat{u}^\varepsilon (t, y) = 0,& (t, y)\in [0, T)\times \hat{D},	\vspace{1mm}\\
 	\hat{u}^\varepsilon (T, y) = \hat{f}(y),& x\in \hat{D},	\\
 	\hat{u}^\varepsilon (t, y) = 0,& (t, y)\in [0, T]\times \partial \hat{D}, 
\end{array}
\right.
\end{eqnarray}
and we obtain the following existence result. 
\begin{thm} \ \label{th_viscosity_general}Assume $[A']$--$[C']$ and $[D]$. 
Then, $u^\varepsilon (t, x)$ is a continuous viscosity solution of $(\ref {eq_PDE})$. 
Moreover, $\hat{u}^\varepsilon (t, y)$ is a continuous viscosity solution of 
satisfying 
\begin{eqnarray}\label{ineq_poly_S}
\sup _{(t, y)\in [0, T]\times \bar{\hat{D}}}|\hat{u}^\varepsilon (t, y)| / (1 + |y|^{2m'}) < \infty . 
\end{eqnarray}
\end{thm}

\begin{proof} 
The latter assertion is by the similar argument to the proof of Proposition \ref {prop_v2} in Appendix \ref{proof:general}. 
Then, the simple calculation gives the former assertion. 
\end{proof}

\begin{rem}
Here we no longer require a local ellpiticity condition [E], because 
we consider viscosity solutions of (\ref {eq_PDE}) and (\ref{eq_PDE_S}) rather than classical solutions:
we can directly show that the function $u^\varepsilon (t, x)$ 
(which is given in the form of a stochastic representation) 
becomes the viscosity solution of the corresponding PDE.

If we further assume a local ellipticity condition such as $[E]$, 
we may show the existence of classical solutions
which is characterized as 
\begin{eqnarray*}
\hat{u}^\varepsilon (t, y) = 
\E \left [\exp \left( -\int ^{T-t}_0c(\pi (Y^{\varepsilon , y}_r), \varepsilon )dr 
\right) \hat{f}(Y^{\varepsilon , y}_{T-t})1_{\{ \tau _{\hat{D}}(Y^{\varepsilon , y}) \geq T-t \}}\right ]. 
\end{eqnarray*}
\end{rem}

Moreover, applying Theorem 8.2 in \cite {Crandall-Lions-Ishii} and Theorem 7.7.2 in \cite{Nagai} to (\ref {eq_PDE_S}),
we have the following uniqueness theorem. 
\begin{thm} \ \label{th_viscosity_uniqueness_S}Assume $[A']$--$[C']$ and $[D]$. 
If $\hat{w}^\varepsilon (t, y)$ is a continuous viscosity solution of $(\ref {eq_PDE_S})$ satisfying the growth condition 
$(\ref {ineq_poly_S})$, then $\hat{u}^\varepsilon = \hat{w}^\varepsilon $. 
\end{thm}

For our generalization of the asymptotic expansion stated as Theorem \ref{th_general}  below,
we need to modify the assumptions $[G]$ and $[H]$ in the previous sections.

In order to state the existence of a function $v^0_k(t, x)$, we prepare the following set
which slightly modifies $\mathcal {H}^{m, p}$ in Definition \ref{def1}. 
Moreover, we define $\hat{\mathcal {G}}^{m}$ similarly to $\mathcal {G}^{m}$, 
replacing $\mathcal {H}^{m, 1}$ and $\mathcal {H}^{m, 2}$ in the definition 
with $\hat{\mathcal {H}}^{m, 1}$ and $\hat{\mathcal {H}}^{m, 2}$, respectively. 
\begin{description}
 \item[\mbox{[G']}] 
\ The condition [G] holds replacing $\mathcal {G}^{m}$ with 
$\hat{\mathcal {G}}^{m}$. That is,
\ $u^0\in \hat{\mathcal {G}}^{m}$, where 
\begin{eqnarray*}
\hat{\mathcal {G}}^{m} &=& 
\Big \{ g\in C^{1, 2}([0, T)\times D)\cap C([0, T]\times \bar{D})\ ; \\&&\hspace{40mm}
\frac{\partial g}{\partial x^i}\in \hat{\mathcal {H}}^{m, 2}, \ 
\frac{\partial ^2g}{\partial x^i\partial x^j}\in \hat{\mathcal {H}}^{m, 1}, \ 
i, j = 1, \ldots , d\Big \},
\end{eqnarray*}
and the set $\hat{\mathcal {H}}^{m, p}$ of $g\in C([0, T)\times \bar{D})$ 
is given by the following:
\end{description}
\begin{defn}
The set $\hat{\mathcal {H}}^{m, p}$ of $g\in C([0, T)\times \bar{D})$ 
is defined to satisfy the following condition: 
There is some $M^g\in C([0, T))\cap L^p([0, T), dt)$ 
such that 
\begin{eqnarray}\label{def_H_map_growth2}
|g(t, x)| \leq M^g(t)(1 + |\iota (x)|^{2m}), \ \ t\in [0, T), \ x, y\in \bar{D}.
\end{eqnarray}
\end{defn}

Accordingly, the condition $[H]$ is replaced by the following:
\begin{description}
 \item[\mbox{[H']}] 
\ The condition $[H]$ holds replacing $\mathcal {G}^{m}$ with 
$\hat{\mathcal {G}}^{m}$:
\ It holds that $v^0_k\in \hat{\mathcal {G}}^{m_n}$, $k = 1, \ldots , n - 1$ for some $m_n\in \Bbb {N}$. 
\end{description}

Then, we obtain the 
generalization of Theorem \ref {th_main} whose proof is given in 
Appendix \ref{proof:general}.
\begin{thm} \ \label{th_general} 
Assume $[A']$--$[C']$, $[D]$, $[F]$ and $[G']$--$[H']$. 
Then, there are positive constants $C_n$ and $\tilde{m}_n$ which are independent of $\varepsilon $ such that 
\begin{eqnarray*}
\left |u^\varepsilon (t, x) - (u^0(t, x) + \sum ^{n-1}_{k = 1}\varepsilon ^kv^0_k(t, x))\right | \leq C_n(1 + |\iota (x)|^{2\tilde{m}_n})\varepsilon ^n, \ \ 
(t, x)\in [0, T]\times \bar{D}. 
\end{eqnarray*}
\end{thm}

\subsection{Proof of Theorem \ref {th_general}}\label{proof:general}
Let $v^\varepsilon _n$ and $\tilde{v}^\varepsilon _n$ be as in Section \ref {sec_proof}. 
Thanks to the assumption $I\subset [0, \infty )$ and $[G']$--$[H']$, 
we can apply similar argument to the proof of Proposition \ref {prop_v}, which tells us that 
$v^\varepsilon _n$ is a viscosity solution of (\ref {eq_PDE2_eps}). 
That is, we obtain the next proposition.

\begin{prop} \ \label{prop_v2}
The function $\tilde{v}^\varepsilon _n$ is a continuous viscosity solution of $(\ref {eq_PDE2_eps})$. 
\end{prop}
\begin{proof}
Until the end of the proof we suppress $\varepsilon $ in the notation. 
First, we check the continuity. 
By the similar argument to the proof of Lemma 4.2 in \cite{Rubio}, we see that 
$v_n$ is continuous on $[0, T)\times \bar{D}$. 
Moreover, similarly to (\ref {poly_v_eps}), 
we see that there are a function $M_n\in C([0, T))\cap L^1([0, T), dt)$ and 
constants $\tilde{C}_n, \tilde{m}_n > 0$ such that 
\begin{eqnarray}\label{moment_tilde_vn}
|\tilde{v}_n(t, x)|\leq \tilde{C}_n\int ^T_tM_n(r)dr(1 + |\iota (x)|^{2\tilde{m}_n}). 
\end{eqnarray}
Thus we get 
\begin{eqnarray*}
\sup _{x\in K\cap \bar{D}}|\tilde{v}_n(t, x)| \leq 
C'_n(1 + \sup _{x\in K}|\iota (x)|^{2m})\left\{ \int ^T_0M_n(r)dr - \int ^t_0M_n(r)dr\right\} 
\ \longrightarrow \ 0, \ \ t\rightarrow T 
\end{eqnarray*}
for any compact set $K\subset \Bbb {R}^d$. 
Thus, $v_n$ is continuous on $[0, T]\times \bar{D}$. 

Next, we show that $v_n$ is a viscosity subsolution of $(\ref {eq_PDE2_eps})$. 
By the definition of $\tilde{v}_n$, we easily get 
$\tilde{v}_n(T, x) = 0$ for $x\in D$ and 
$\tilde{v}_n(t, x) = 0$ for $(t, x)\in [0, T]\times \partial D$. 
Now take any $(t, x)\in [0, T)\times D$ and 
let $\varphi $ be $C^{1, 2}$-function such that $v_n - \varphi $ has a maximum $0$ at $(t, x)$. 
We may assume that $\varphi $ and its derivatives have polynomial growth rates in $x$ uniformly in $t$. 
By the Markov property, we have 
\begin{eqnarray*}
&&\E \left [ J(h\wedge \tau _D(X^x))\tilde{v}_n
\left (t + h\wedge \tau _D(X^x), X^x_{h\wedge \tau _D(X^x)}\right )\right ] \ = \ 
\E \left [ J(h)\tilde{v}_n(t + h, X^x_h)1_{\{ \tau _D(X^x)\geq h\} }\right ]\\
&=& 
\E \left [\int ^{(T - t)\wedge (\tau _D(X^x_{\cdot + h}) + h)}_h
J(r)g_n(t + r, X^x_r)dr1_{\{ \tau _D(X^x)\geq h\} }\right ]
\end{eqnarray*}
for $h\in (0, T - t)$, 
where $J(r) = \exp \left ( -\int ^r_0c(X^x_v, \varepsilon )dv\right )$. 
Since $\tau _D(X^x_{\cdot + h}) = \tau _D(X^x) - h$ on $\{\tau _D(X^x) \geq h\}$, we obtain 
\begin{eqnarray*}
\E \left [ J(h\wedge \tau _D(X^x))\tilde{v}_n
\left (t + h\wedge \tau _D(X^x), X^x_{h\wedge \tau _D(X^x)}\right )\right ] = 
\tilde{v}_n(t, x) - \E \left [\int ^{h\wedge \tau _D(X^x)}_0J(r)g_n(t + r, X^x_r)dr\right ]. 
\end{eqnarray*}
Therefore, 
\begin{eqnarray*}
&&\varphi (t, x) \ =\  \tilde{v}_n(t, x)\\
&=& 
\E \left [ J(h\wedge \tau _D(X^x))\tilde{v}_n
\left (t + h\wedge \tau _D(X^x), X^x_{h\wedge \tau _D(X^x)}\right )\right ] + 
\E \l[\int ^{h\wedge \tau _D(X^x)}_0g_n(t + r, X^x_r)dr\r]\\
&\leq & 
\E \left [ J(h\wedge \tau _D(X^x))\varphi 
\left (t + h\wedge \tau _D(X^x), X^x_{h\wedge \tau _D(X^x)}\right )\right ] + 
\E \l[\int ^{h\wedge \tau _D(X^x)}_0g_n(t + r, X^x_r)dr\r]. 
\end{eqnarray*}
Note that $[A']$, $[C']$ and (\ref {gen_linear_growth}) imply that 
\begin{eqnarray*}
\int ^\cdot _0J(r)\sigma ^{ij}(X^x_r, \varepsilon )\frac{\partial }{\partial x^i}\varphi (t + r, X^x_r)dB^i_r
\end{eqnarray*}
is a martingale. 
Thus, applying Ito's formula, we get 
\begin{eqnarray*}
-\frac{1}{h}\int ^h_0\E 
\left[\left\{ \left( \frac{\partial }{\partial t} + \mathscr {L}\right )
\varphi (t + r, X^x_r) + g_n(t + r, X^x_r)\right\} 1_{\{ \tau _D(X^x)\geq h\} }\right ]dr \leq 0. 
\end{eqnarray*}
Letting $h\rightarrow 0$, we see that 
\begin{eqnarray*}
-\frac{\partial }{\partial t}\varphi (t, x) - \mathscr {L}\varphi (t, x) - g_n(t, x) \leq 0. 
\end{eqnarray*}
Hence, $\tilde{v}_n$ is a viscosity subsolution of (\ref {eq_PDE2_eps}). 
By the similar argument, we also find that $\tilde{v}_n$ is a viscosity supersolution. 
\end{proof}

To see the equivalence $v^\varepsilon _n = \tilde{v}^\varepsilon _n$, 
we need to give a new proof of Proposition \ref {prop_uniqueness_v_eps} under the assumptions of Theorem \ref {th_general}. 

\begin{proof}[Proof of Proposition \ref {prop_uniqueness_v_eps}] 
Set 
$\bar{u}^\varepsilon _n(t, x) = 
u^0(t, x) + \sum ^{n-1}_{k = 1}\varepsilon ^kv^0_k(t, x) + \varepsilon ^n\tilde{v}^\varepsilon _n(t, x)$. 
The analogous argument of the proof of Proposition \ref {prop_v} implies that 
$\bar{u}^\varepsilon _n$ is the continuous viscosity solution of (\ref {eq_PDE}). 
Thus $\bar{u}^\varepsilon _n(t, \pi (y))$ is the continuous viscosity solution of (\ref {eq_PDE_S}). 
By (\ref {moment_tilde_vn}), we see that 
$\bar{u}^\varepsilon _n(t, \pi (y))$ has a polynomial growth rate in $y$ uniformly in $t$. 
Then, Theorem \ref {th_viscosity_uniqueness_S} leads us to 
$\bar{u}^\varepsilon _n(t, \pi (y)) = \hat{u}^\varepsilon (t, y)$, which implies 
$\bar{u}^\varepsilon _n(t, x) = u^\varepsilon (t, x)$. 
This equality, (\ref {temp_expansion2}) and (\ref {moment_tilde_vn}) imply the assertion. 
\end{proof}

Now, we obtain the assertion of Theorem \ref {th_general} by the same way as that of Theorem \ref {th_main}.

\end{document}

%% file: tymacro_3.tex

\def\bn{{\bf n}}
\def\A{{\bf A}}
\def\B{{\bf B}}
\def\C{{\bf C}}
\def\D{{\bf D}}
\def\E{{\bf E}}
\def\F{{\bf F}}
\def\G{{\bf G}}
\def\H{{\bf H}}
\def\I{{\bf I}}
\def\J{{\bf J}}
\def\K{{\bf K}}
\def\L{{\bf L}}
\def\M{{\bf M}}
\def\N{{\bf N}}
\def\O{{\bf O}}
\def\P{{\bf P}}
\def\Q{{\bf Q}}
\def\R{{\bf R}}
\def\S{{\bf S}}
\def\T{{\bf T}}
\def\U{{\bf U}}
\def\V{{\bf V}}
\def\W{{\bf W}}
\def\X{{\bf X}}
\def\Y{{\bf Y}}
\def\Z{{\bf Z}}
\def\cala{{\cal A}}
\def\calb{{\cal B}}
\def\calc{{\cal C}}
\def\cald{{\cal D}}
\def\cale{{\cal E}}
\def\calf{{\cal F}}
\def\calg{{\cal G}}
\def\calh{{\cal H}}
\def\cali{{\cal I}}
\def\calj{{\cal J}}
\def\calk{{\cal K}}
\def\call{{\cal L}}
\def\calm{{\cal M}}
\def\caln{{\cal N}}
\def\calo{{\cal O}}
\def\calp{{\cal P}}
\def\calq{{\cal Q}}
\def\calr{{\cal R}}
\def\cals{{\cal S}}
\def\calt{{\cal T}}
\def\calu{{\cal U}}
\def\calv{{\cal V}}
\def\calw{{\cal W}}
\def\calx{{\cal X}}
\def\caly{{\cal Y}}
\def\calz{{\cal Z}}
%
\def\sskip{\hspace{0.5cm}}
\def\simleq{ \raisebox{-.7ex}{\em $\stackrel{{\textstyle <}}{\sim}$} }
\def\leqsim{ \raisebox{-.7ex}{\em $\stackrel{{\textstyle <}}{\sim}$} }
\def\ep{\epsilon}
\def\half{\frac{1}{2}}
\def\iku{\rightarrow}
\def\Iku{\Rightarrow}
\def\ikup{\rightarrow^{p}}
\def\inclusion{\hookrightarrow}
\def\cadlag{c\`adl\`ag\ }
\def\up{\uparrow}
\def\down{\downarrow}
\def\doti{\Leftrightarrow}
\def\douti{\Leftrightarrow}
\def\dochi{\Leftrightarrow}
\def\douchi{\Leftrightarrow}%
\def\yy{\\ && \nonumber \\}
\def\y{\vspace*{3mm}\\}
\def\nn{\nonumber}
\def\be{\begin{equation}}
\def\ee{\end{equation}}
\def\bea{\begin{eqnarray}}
\def\eea{\end{eqnarray}}
\def\beas{\begin{eqnarray*}}
\def\eeas{\end{eqnarray*}}
%
\def\hd{\hat{D}}
\def\hv{\hat{V}}
\def\hsd{{\hat{d}}}
\def\hx{\hat{X}}
\def\hsx{\hat{x}}
\def\bsx{\bar{x}}
\def\bsd{{\bar{d}}}
\def\bx{\bar{X}}
\def\ba{\bar{A}}
\def\bb{\bar{B}}
\def\bc{\bar{C}}
\def\bv{\bar{V}}
\def\balpha{\bar{\alpha}}
\def\bbalpha{\bar{\bar{\alpha}}}
\def\combi{\l(\begin{array}{c}\alpha\\ \beta \end{array}\r)}
\def\f{^{(1)}}
\def\s{^{(2)}}
\def\ss{^{(2)*}}
\def\l{\left}
\def\r{\right}
\def\a{\alpha}
\def\b{\beta}
\def\L{\Lambda}